\documentclass[runningheads]{llncs}

\usepackage[T1]{fontenc}
\usepackage{graphicx}
\usepackage{color}
\usepackage{algorithm}
\usepackage{comment}
\usepackage[noend]{algpseudocode}
\usepackage{amsmath}
\usepackage{hyperref}
\usepackage{enumitem}
\usepackage{booktabs}
\usepackage{multirow}
\usepackage{amssymb}

\usepackage[table,xcdraw]{xcolor}
\usepackage[table]{xcolor}

\algnewcommand{\LineComment}[1]{\State \(\triangleright\) \textbf{#1}}

\definecolor{lightgray}{gray}{0.9}

\begin{document}
 
\title{Fast Simulation Algorithms for OLH using Binomial Modeling}
 
\titlerunning{Fast Simulation Algorithms for OLH using Binomial Modeling}
 
\author{Berkay Kemal Balioglu \and Alireza Khodaie \and \\ M. Emre Gursoy}
\authorrunning{Balioglu, Khodaie, and Gursoy}

\institute{Department of Computer Engineering, Koç University, Istanbul, Turkey\\
\email{\{bbalioglu23, akhodaie22, emregursoy\}@ku.edu.tr}}
 
\maketitle

\begin{abstract}
Optimized Local Hashing (OLH) is a widely used hash-based Local Differential Privacy (LDP) protocol, and simulation-based experimentation is the standard approach for evaluating OLH and OLH-based applications in research. However, the existing OLH simulations have $O(nd)$ computational complexity, where $n$ is the user population size and $d$ is the domain size, and can lead to significant execution times as $n$ and $d$ grow. In this paper, we propose two fast simulation algorithms for OLH (2-Binom and 3-Binom) grounded in Binomial modeling. Our key insight is that, for any domain value $v$, the total number of users whose perturbed reports support $v$ can be decomposed into a sum of two or three Binomial random variables. Using this insight, our algorithms reduce the simulation complexity to $O(n + d)$ without hurting statistical equivalence. In particular, we theoretically prove that both algorithms yield unbiased frequency estimations with variances identical to those of the original OLH simulations. Experiments on real-world datasets confirm that both approaches reduce execution times from several minutes to milliseconds, yielding significant speedups with no change in utility.

\keywords{Local differential privacy \and optimized local hashing \and privacy-enhancing technologies \and privacy protocols.}
\end{abstract}

\section{Introduction} \label{sec:introduction}

The growing adoption of data-driven services has increased concerns around user privacy, making privacy-preserving data collection an increasingly critical challenge. Local differential privacy (LDP) has emerged as a response to this challenge, offering mathematically rigorous guarantees that protect individual users without relying on a trusted server \cite{cormode2018privacy,erlingsson2014rappor,wang2017locally,yang2023local}. Optimized Local Hashing (OLH) \cite{wang2017locally} is among one of the most widely adopted LDP protocols -- in OLH, each user encodes their true value by applying a randomly drawn hash function before perturbation. This hash-based design makes OLH particularly attractive for large domains; as a result, it has been used extensively in subsequent works and downstream tasks \cite{cormode2019answering,qian2023collaborative,wang2018locally,wang2019locally,xu2023mlpkv}, and it has been implemented in several LDP benchmarking platforms \cite{arcolezi2022multi,cormode2021frequency,khodaie2025postprocessing}. 

Evaluating LDP protocols and applications built on top of them requires simulation-based experiments. A researcher wishing to assess a new application, e.g., one that uses OLH as a building block, typically simulates the entire process on a single machine by iterating through each of the $n$ users one by one, accumulating their perturbed outputs, and finally performing server-side estimation to obtain frequency estimates $\hat{f}_v$ for all values $v$ in domain $\mathcal{D}$. This approach accurately reflects real-world deployments but carries a computational burden that grows with both the number of users $n$ and the domain size $d$. For OLH specifically, there is an $O(nd)$ time complexity, which can cause simulation times of several minutes on moderately sized datasets. 

In this paper, we propose two fast simulation algorithms for OLH (called 2-Binom and 3-Binom) that replace the costly user-by-user simulation with draws from a small number of Binomial random variables. Our key insight is that, for any domain value $v \in \mathcal{D}$, users can be partitioned into two groups based on whether their true value equals $v$ or not. The contribution of the first group to the support of $v$ ($\mathrm{Sup}(v)$) follows a Binomial distribution. For the second group, although each user can contribute to $\mathrm{Sup}(v)$ through two ways (a hash collision between their true value and $v$, or a perturbation to $H_u(v)$) these two sub-cases can be algebraically combined into one Binomial random variable. This yields our 2-Binom algorithm, which reduces the $O(nd)$ complexity to $O(n+d)$, or even $O(d)$ if per-value counts are precomputed and cached. We further propose an alternative algorithm (3-Binom) that keeps the two sub-cases in the second group separate, modeling the number of colliding users as an intermediate Binomial draw before conditioning the remaining perturbation outcomes on it. 

We prove that both 2-Binom and 3-Binom yield unbiased frequency estimates with variance identical to that of standard OLH simulation, and that the two algorithms are statistically equivalent to each other despite their different decomposition structures. We also perform experiments on four real-world datasets across a range of privacy budgets and metrics. Results show that 2-Binom and 3-Binom reduce empirical simulation times from several minutes to milliseconds, with no change in the correctness or utility of the simulation results.

In short, our main contributions are as follows:
\vspace{-4pt}
\begin{itemize}
\item We propose two simulation algorithms (2-Binom and 3-Binom) for OLH, in which the collective behavior of $n$ users is captured by drawing random samples from appropriate Binomial random variables.
\item We formally prove that both 2-Binom and 3-Binom algorithms yield unbiased frequency estimates with variances equal to the original OLH simulations.
\item We empirically validate 2-Binom and 3-Binom on multiple real-world datasets, demonstrating orders-of-magnitude reduction in execution time while producing estimates that are statistically indistinguishable from those of the existing OLH simulation.
\end{itemize}

\vspace{-8pt}

\section{Background and Preliminaries}

\subsection{Data Model and Notation}

We denote the population of users by $\mathcal{P}$, and the number of users in the population by $n = |\mathcal{P}|$. Each user $u \in \mathcal{P}$ holds a single private value $v_u$ drawn from a discrete domain $\mathcal{D}$, and we denote the cardinality of the domain by $d = |\mathcal{D}|$. For any value $v \in \mathcal{D}$, we denote by $n_v$ the number of users whose true value equals $v$, i.e., $n_v = |\{u \in \mathcal{P} : v_u = v\}|$, and by $f_v = n_v / n$ the true frequency of $v$ in the population. Note that $\sum_{v \in \mathcal{D}} n_v = n$ and $\sum_{v \in \mathcal{D}} f_v = 1$. 
 
\subsection{Local Differential Privacy (LDP)}

LDP is a widely adopted framework for collecting user data while preserving individual privacy. In a typical LDP setting, a server (data collector) wishes to gather statistics from a population of users without learning any individual user's true value $v_u$. To protect privacy, each user applies a randomized mechanism $\psi$ to their true value $v_u$ locally on their own device, and transmits the perturbed result to the server. Upon receiving perturbed reports from the entire population, the server applies estimation techniques to recover aggregate statistics. 
 
\begin{definition}[$\varepsilon$-LDP]
\label{def:LDP}
A randomized mechanism $\psi$ satisfies $\varepsilon$-local differential privacy ($\varepsilon$-LDP) if and only if, for any two input values $v_1, v_2 \in \mathcal{D}$:
\begin{equation}
\label{eq:LDP}
\forall\, y \in \mathrm{Range}(\psi): \quad \frac{\Pr[\psi(v_1) = y]}{\Pr[\psi(v_2) = y]} \leq e^{\varepsilon},
\end{equation}
where $\mathrm{Range}(\psi)$ denotes the set of all possible outputs of $\psi$.
\end{definition}
 
Intuitively, $\varepsilon$-LDP ensures that an observer of the perturbed output $y$ cannot distinguish between any two true values $v_1$ and $v_2$ with a likelihood ratio exceeding $e^{\varepsilon}$. The privacy budget $\varepsilon$ controls the strength of privacy: smaller $\varepsilon$ yields stronger privacy protection, while larger values permit higher utility.
 
\subsection{Optimized Local Hashing (OLH)}
\label{sec:olh}

OLH \cite{wang2017locally} is a hash-based LDP protocol in which each user encodes their true value by applying a randomly chosen hash function that maps domain elements into an integer range $\{0, 1, \ldots, g-1\}$ of size $g \geq 2$. The parameter $g$ controls the trade-off between encoding resolution and perturbation likelihood. The choice of $g$ that minimizes estimation variance was shown to be $g = e^{\varepsilon} + 1$~\cite{wang2017locally}.
 
Let $\mathcal{H}$ denote a universal family of hash functions, where each $H \in \mathcal{H}$ satisfies $H: \mathcal{D} \rightarrow \{0, 1, \ldots, g-1\}$. The OLH protocol proceeds as follows.
 
\textbf{Encoding and Perturbation.} Each user $u \in \mathcal{P}$ independently and uniformly samples a hash function $H_u$ from $\mathcal{H}$ and computes the hash of their true value: $x_u = H_u(v_u)$. The hashed value $x_u$ is then perturbed by reporting $x'_u \in \{0, 1, \ldots, g-1\}$ according to:
\begin{equation}
\label{eq:olh_perturb}
\Pr[x'_u = i] = \begin{cases} p = \frac{e^{\varepsilon}}{e^{\varepsilon} + g - 1}, & \text{if } i = x_u \\[3pt] q = \frac{1}{e^{\varepsilon} + g - 1}, & \text{otherwise} \end{cases}
\end{equation}
The user sends the pair $\langle H_u,\, x'_u \rangle$ to the server.
 
\textbf{Server-side Estimation.} Upon receiving reports from all users, the server estimates the frequency of each value $v \in \mathcal{D}$. For a given $v$, the server first computes the support of $v$:
\begin{equation}
\label{eq:olh_sup}
\mathrm{Sup}(v) = \bigl|\{ u \in \mathcal{P} : x'_u = H_u(v) \}\bigr|
\end{equation}
which counts the number of users whose reported hash output is consistent with $v$ under their hash function. The estimated frequency $\hat{f}_v$ is then obtained by:
\begin{equation}
\label{eq:olh_est}
\hat{f}_v = \frac{(e^{\varepsilon} + g - 1)\bigl(g\cdot \mathrm{Sup}(v) - n\bigr)}{(e^{\varepsilon} - 1)(g - 1)\,n}
\end{equation}

\subsection{Simulating OLH: Existing Approach}
\label{sec:simulation}

To empirically evaluate LDP applications, including those involving OLH, researchers rely on simulation-based experiments, i.e., the encoding, perturbation, and estimation steps of all users and the server are executed on a single machine. The representative pseudocode for simulating OLH is given in Algorithm \ref{algo:olh}. The simulation begins by initializing the values of parameters $g$ and $p$, as well as an empty list of reports $L$, which will store the perturbed $\langle H_u,\, x'_u \rangle$ pairs of each user. Then, in the user-side encoding and perturbation phase (lines 4-12), the simulation iterates through each user $u_i \in \mathcal{P}$ one by one. For each user, a random seed $s_{u_i}$ is sampled to determine $H_{u_i}$ (this is the default strategy to sample each user's hash function $H_{u_i}$ randomly from $\mathcal{H}$ in a way that conforms with typical hash libraries in Python, e.g., \texttt{xxhash}). The user's true value $v_{u_i}$ is hashed to produce $x_{u_i}$, and the hashed value is perturbed according to Equation \ref{eq:olh_perturb} to determine $x'_{u_i}$. The perturbed pair $\langle x'_{u_i}, H_{u_i} \rangle$ is stored in $L$ (line 12). After all users have been simulated, the server-side estimation phase begins (lines 13-18). For each value $v \in \mathcal{D}$ to be estimated, its $\mathrm{Sup}(v)$ is computed by iterating through all stored reports in $L$ and counting those whose perturbed hash output is consistent with $v$ under the reported hash function (lines 16-17). Then, $\mathrm{Sup}(v)$ is fed into the OLH estimator to compute $\hat{f}_v$.

\begin{algorithm}[!t]
\caption{Existing OLH Simulation}
\label{algo:olh}
\begin{algorithmic}[1]
    \Statex \textbf{Input:} Domain $\mathcal{D}$, budget $\varepsilon$, values of users in $\mathcal{P}$ $(v_{u_1}, v_{u_2}, \ldots, v_{u_n})$
    \Statex \textbf{Output:} Estimated frequencies $\hat{f}_v$ for all $v \in \mathcal{D}$
    \State Compute $g \leftarrow \lfloor e^{\varepsilon} + 1 \rfloor$,~~ $p \leftarrow \frac{e^{\varepsilon}}{e^{\varepsilon} + g - 1}$
    \State Initialize empty list of reports $L$
    \LineComment{User-side encoding and perturbation}
    \For{each user $u_i \in \mathcal{P}$}
        \State Sample a random seed $s_{u_i}$ and initialize $H_{u_i}$ accordingly
        \State $x_{u_i} \leftarrow H_{u_i}(v_{u_i}) \bmod g$
        \State $r \leftarrow$ sample a random float in $[0,1]$
        \If{$r \leq p$}
            \State $x'_{u_i} \leftarrow x_{u_i}$
        \Else
            \State $x'_{u_i} \leftarrow$ sample uniformly from $\{0,\ldots,g-1\} \setminus \{x_{u_i}\}$
        \EndIf
        \State Store $\langle x'_{u_i},\, H_{u_i} \rangle$ in $L$
    \EndFor
    \LineComment{Server-side estimation}
    \For{each value $v \in \mathcal{D}$}
        \State Initialize $\mathrm{Sup}(v) \leftarrow 0$
        \For{each $\langle x'_{u_i},\, H_{u_i} \rangle$ in $L$}
            \State Increment $\mathrm{Sup}(v)$ by +1 iff $x'_{u_i} = H_{u_i}(v) \bmod g $
        \EndFor
        \State Compute $\hat{f}_v \leftarrow\frac{(e^{\varepsilon} + g - 1)(g\cdot \mathrm{Sup}(v) - n)}{(e^{\varepsilon} - 1)(g - 1)\,n}$  
    \EndFor
    \State \Return $\hat{f}_v$ for all $v \in \mathcal{D}$
\end{algorithmic}
\label{alg:olh}
\end{algorithm}

The overall complexity of Algorithm \ref{algo:olh} is $O(nd)$. The user-side phase is $O(n)$ since it iterates over all $n$ users, but the server-side phase dominates the total runtime with $O(nd)$ complexity because it needs to compute $\mathrm{Sup}(v)$ for all $d$ values. Computing $\mathrm{Sup}(v)$ requires iterating over the list $L$, which contains $n$ elements. In practice, this dominates the total runtime because of the many hash computations that need to be performed (line 17 is performed $O(nd)$ times). As $n$ and $d$ grow -- for instance, $n$ in the order of millions and $d$ in the order of hundreds -- this $O(nd)$ cost causes simulation times to grow. Addressing this scalability problem is the primary motivation of our paper.

\section{Proposed Methodology}
 
\subsection{Mathematical Intuition}
\label{sec:insight}

We first analyze how the collective effect of OLH can be modeled using Binomial random variables. For value $v \in \mathcal{D}$, we observe that each user $u_i \in \mathcal{P}$ independently contributes either 0 or 1 to $\mathrm{Sup}(v)$, and we partition the $n$ users into two groups based on their true value:

\textbf{Group 1: Users with true value $v$ ($n_v$ users).} For a user $u_i$ with $v_{u_i} = v$, their hashed value is $x_{u_i} = H_{u_i}(v_{u_i}) \bmod g = H_{u_i}(v) \bmod g$. This user contributes +1 to $\mathrm{Sup}(v)$ if and only if the perturbed output satisfies $x'_{u_i} = x_{u_i}$, which by the perturbation rule in Eq.~\ref{eq:olh_perturb} occurs with probability $p$. Therefore, each user in Group 1 contributes an independent Bernoulli($p$) random variable to $\mathrm{Sup}(v)$, and summing over all $n_v$ such users yields:
\begin{equation}
\label{eq:group1}
    X \sim \mathrm{Binomial}(n_v,\, p)
\end{equation}

\noindent\textbf{Group 2: Users with true value $\neq v$ ($n - n_v$ users).} For a user $u_i$ with $v_{u_i} \neq v$, their hashed value is $x_{u_i} = H_{u_i}(v_{u_i}) \bmod g$. This user contributes +1 to $\mathrm{Sup}(v)$ if and only if $x'_{u_i} = H_{u_i}(v) \bmod g$, which can occur via two mutually exclusive sub-cases:

\textbf{Sub-case A:} Hash function $H_{u_i}$ causes a hash collision between $v_{u_i}$ and $v$, i.e., $H_{u_i}(v_{u_i}) \bmod g = H_{u_i}(v) \bmod g$, and also, the perturbation in Eq.~\ref{eq:olh_perturb} retains this value. By the universality of $\mathcal{H}$, the probability of such a hash collision is $\frac{1}{g}$, and the perturbation retains the value with probability $p$. The probability of this sub-case is therefore $\frac{p}{g}$. 

\textbf{Sub-case B:} There is no hash collision, i.e., $H_{u_i}(v_{u_i}) \bmod g \neq H_{u_i}(v) \bmod g$, but the perturbation in Eq.~\ref{eq:olh_perturb} flips $x_{u_i}$ to specifically become $x'_{u_i} = H_{u_i}(v) \bmod g$. The probability of no hash collision is $\frac{g-1}{g}$, and the perturbation outputs a specific other value with probability $q$. The probability of this sub-case is therefore $\frac{(g-1)q}{g}$.

Combining the two sub-cases and substituting $p = \frac{e^{\varepsilon}}{e^{\varepsilon}+g-1}$ and $q = \frac{1}{e^{\varepsilon}+g-1}$, the total probability that a user in Group 2 contributes +1 to $\mathrm{Sup}(v)$ is:
\begin{equation}
    \frac{p}{g} + \frac{(g-1)q}{g} = \frac{p + (g-1)q}{g} = \frac{\frac{e^{\varepsilon}}{e^{\varepsilon}+g-1} + \frac{g-1}{e^{\varepsilon}+g-1}}{g} = \frac{e^{\varepsilon} + g - 1}{g(e^{\varepsilon}+g-1)} = \frac{1}{g}
\end{equation}
The two sub-cases therefore combine into a single Bernoulli$(\frac{1}{g})$ per user, and summing over all $n - n_v$ users in Group 2 yields:
\begin{equation}
    Y \sim \mathrm{Binomial}\!\left(n - n_v,\, \frac{1}{g}\right)
\end{equation}

\textbf{Finally,} since $X$ and $Y$ are independent, combining Groups 1 and 2, the total support for any value $v \in \mathcal{D}$ can be expressed as:
\begin{equation}
\label{eq:decomp}
    \mathrm{Sup}(v) = X + Y, \text{where } X \sim \mathrm{Binomial}(n_v,\, p) \text{, } Y \sim \mathrm{Binomial}\!\left(n - n_v,\, \frac{1}{g}\right)
\end{equation}
This decomposition is the key insight underlying our proposed methodology. Rather than simulating all $n$ users individually to compute $\mathrm{Sup}(v)$, we only need to draw two Binomial random variables per domain value $v \in \mathcal{D}$. The only inputs required are $n$, $n_v$, $p$, $q$, and $g$, all of which are either determined by the dataset or by the protocol parameters. In the next subsection, we show how this insight translates into a fast simulation algorithm.

\subsection{Proposed Simulation Algorithm (2-Binom)}

We now translate this mathematical insight into a fast simulation algorithm for OLH. Our proposed 2-Binom simulation algorithm is given in Algorithm \ref{alg:2binom}. 2-Binom replaces the user-by-user simulation with draws from two Binomial random variables. It proceeds in three phases: First, it computes $g$ and $p$. Second, between lines 2-4, it iterates through the user population once to compute the per-value counts $n_v$ for all $v \in \mathcal{D}$ and the total population size $n$. Third, between lines 5-10, it iterates over each value $v \in \mathcal{D}$ and applies the Binomial approach from Eq.~\ref{eq:decomp}. For each $v$, two independent Binomial random variables are created and samples are drawn from them: $X \sim \mathrm{Binomial}(n_v, p)$ and $Y \sim \mathrm{Binomial}(n - n_v, \frac{1}{g})$. The sum of the drawn samples gives $\mathrm{Sup}(v)$, from which the estimated frequency $\hat{f}_v$ is computed via the estimator on line 10.

\begin{algorithm}[!t]
\caption{Proposed 2-Binom Simulation Algorithm} \label{alg:2binom}
\begin{algorithmic}[1]
    \Statex \textbf{Input:} Domain $\mathcal{D}$, budget $\varepsilon$, values of users in $\mathcal{P}$ $(v_{u_1}, v_{u_2}, \ldots, v_{u_n})$
    \Statex \textbf{Output:} Estimated frequencies $\hat{f}_v$ for all $v \in \mathcal{D}$
    \State Compute $g \leftarrow \lfloor e^{\varepsilon} + 1 \rfloor$,~~ $p \leftarrow \frac{e^{\varepsilon}}{e^{\varepsilon} + g - 1}$
    \LineComment{Find $n_v$ and $n$ from user values}
    \For{each user $u_i \in \mathcal{P}$}
        \State $n_{v_{u_i}} \leftarrow n_{v_{u_i}} + 1$, \quad $n \leftarrow n + 1$
    \EndFor
    \LineComment{Binomial approach}
    \For{each value $v \in \mathcal{D}$}
        \State $x \leftarrow$ draw a random sample from $\mathrm{Binomial}(n_v,\, p)$
        \State $y \leftarrow$ draw a random sample from $\mathrm{Binomial}\!\left(n - n_v,\, \frac{1}{g}\right)$
        \State $\mathrm{Sup}(v) \leftarrow x + y$
        \State $\hat{f}_v \leftarrow \frac{\mathrm{Sup}(v) - (n/g)}{n \cdot (p - (1/g))}$
    \EndFor
    \State \Return $\hat{f}_v$ for all $v \in \mathcal{D}$
\end{algorithmic}
\end{algorithm}

The time complexity of 2-Binom is $O(n + d)$. Computing $g$ and $p$ are $O(1)$. Iterating over all $n$ users to compute the per-value counts $n_v$ is $O(n)$. Iterating over all $v \in \mathcal{D}$ to compute $\hat{f}_v$ using the Binomial approach is $O(d)$ because the loop iterates $d$ times, and the computations performed within each iteration are $O(1)$. Therefore, the overall complexity of Algorithm \ref{alg:2binom} is $O(n + d)$. Note that this is lower than the $O(nd)$ complexity of Algorithm \ref{algo:olh}. 

The complexity of Algorithm \ref{alg:2binom} can be further reduced to $O(d)$ with a simple but practically important observation. In many LDP experiments, the same dataset is used across multiple simulation runs, e.g., when experiments are repeated to achieve statistical significance, or when $\varepsilon$ is varied across a range of values on a fixed dataset. In such cases, the per-value counts $n_v$ and the population size $n$ are identical across all runs. Therefore, lines 2-4 of Algorithm \ref{alg:2binom} need to be executed only once, and its result can be stored and reused. Subsequent simulation runs then only execute the last phase (lines 5-10), which has complexity $O(d)$. This means that after a one-time $O(n)$ pre-processing cost, an arbitrary number of simulation repetitions can each be completed in $O(d)$ time.

\subsection{Utility Analysis of 2-Binom}
\label{sec:utility}

To formally establish that Algorithm \ref{alg:2binom} produces correct results (i.e., equivalent to existing OLH simulations), we prove two properties: (i) estimations are unbiased, and (ii) estimation variances are identical to that of Algorithm \ref{algo:olh}. 

\begin{theorem}[Unbiasedness]
\label{thm:unbiased}
The frequency estimations $\hat{f}_v$ produced by Algorithm~2 are unbiased, i.e., $\mathbb{E}[\hat{f}_v] = f_v$ for all $v \in \mathcal{D}$.
\end{theorem}

\begin{proof}
Applying the expectation to the estimator on line 10 of Algorithm \ref{alg:2binom}:
\begin{equation}
\mathbb{E}[\hat{f}_v] = \mathbb{E}\!\left[\frac{\mathrm{Sup}(v) - n/g}{n \cdot (p - 1/g)}\right] = \frac{\mathbb{E}[\mathrm{Sup}(v)] - n/g}{n \cdot (p - 1/g)}
\end{equation}
From Eq.~\ref{eq:decomp}, we have $\mathrm{Sup}(v) = X + Y$ where $X \sim \mathrm{Binomial}(n_v, p)$ and $Y \sim \mathrm{Binomial}(n - n_v, 1/g)$. Plugging these into the above equation:
\begin{equation}
\mathbb{E}[\mathrm{Sup}(v)] = \mathbb{E}[X] + \mathbb{E}[Y] = n_v \cdot p + (n - n_v) \cdot \frac{1}{g}
\end{equation}
Substituting into the estimator:
\begin{equation}
\mathbb{E}[\hat{f}_v] = \frac{n_v \cdot p + (n - n_v) \cdot \frac{1}{g} - \frac{n}{g}}{n \cdot (p - \frac{1}{g})} = \frac{n_v \cdot p + \frac{n - n_v}{g} - \frac{n}{g}}{n \cdot (p - \frac{1}{g})}
\end{equation}
Simplifying the numerator:
\begin{equation}
n_v \cdot p + \frac{n - n_v}{g} - \frac{n}{g} = n_v \cdot p + \frac{n - n_v - n}{g} = n_v \cdot p - \frac{n_v}{g} = n_v \cdot \left(p - \frac{1}{g}\right)
\end{equation}
Therefore:
\begin{equation}
\mathbb{E}[\hat{f}_v] = \frac{n_v \cdot (p - \frac{1}{g})}{n \cdot (p - \frac{1}{g})} = \frac{n_v}{n} = f_v
\end{equation}
\end{proof}

\begin{theorem}[Variance]
\label{thm:variance}
The variance of the frequency estimations $\hat{f}_v$ produced by Algorithm \ref{alg:2binom} is:
\begin{equation}
\label{eq:variance}
\mathrm{Var}[\hat{f}_v] = \frac{\frac{1}{g}\left(1 - \frac{1}{g}\right)}{n\left(p - \frac{1}{g}\right)^2} + \frac{f_v\left(1 - p - \frac{1}{g}\right)}{n\left(p - \frac{1}{g}\right)}
\end{equation}
\end{theorem}

\begin{proof}
We compute the variance of the estimator on line~10 of Algorithm \ref{alg:2binom}. Since $n$ and $g$ are constants, we have:
\begin{equation}
\mathrm{Var}[\hat{f}_v] = \mathrm{Var}\!\left[\frac{\mathrm{Sup}(v) - n/g}{n \cdot (p - 1/g)}\right] = \frac{\mathrm{Var}[\mathrm{Sup}(v)]}{n^2 \cdot (p - 1/g)^2}
\end{equation}
Since $X$ and $Y$ are independent, the variance of their sum is the sum of their variances:
\begin{equation}
\mathrm{Var}[\mathrm{Sup}(v)] = \mathrm{Var}[X + Y] = \mathrm{Var}[X] + \mathrm{Var}[Y]
\end{equation}
Recalling that $X \sim \mathrm{Binomial}(n_v, p)$ and $Y \sim \mathrm{Binomial}(n - n_v, 1/g)$, we have:
\begin{equation}
\mathrm{Var}[X] = n_v \cdot p \cdot (1 - p), \qquad \mathrm{Var}[Y] = (n - n_v) \cdot \frac{1}{g} \cdot \left(1 - \frac{1}{g}\right) \label{eq:emre}
\end{equation}
Substituting back and separating into two terms by expanding $(n - n_v)$:
\begin{align}
\mathrm{Var}[\hat{f}_v] &= \frac{n_v \cdot p(1-p) + (n-n_v) \cdot \frac{1}{g}\left(1 - \frac{1}{g}\right)}{n^2 \cdot (p - 1/g)^2} \\
&= \frac{n \cdot \frac{1}{g}\left(1 - \frac{1}{g}\right) + n_v \cdot p(1-p) - n_v \cdot \frac{1}{g}\left(1 - \frac{1}{g}\right)}{n^2 \cdot (p - 1/g)^2} \\
&= \frac{\frac{1}{g}\left(1-\frac{1}{g}\right)}{n\left(p-\frac{1}{g}\right)^2} + \frac{n_v\left[p(1-p) - \frac{1}{g}\left(1-\frac{1}{g}\right)\right]}{n^2\left(p-\frac{1}{g}\right)^2}
\end{align}
For the second term, we substitute $n_v = n \cdot f_v$ and simplify the bracket:
\begin{align*}
p(1-p) - \frac{1}{g}\left(1 - \frac{1}{g}\right) &= p - p^2 - \frac{1}{g} + \frac{1}{g^2} \nonumber = \left(p - \frac{1}{g}\right) - \left(p^2 - \frac{1}{g^2}\right) \nonumber \\
&= \left(p - \frac{1}{g}\right) - \left(p - \frac{1}{g}\right)\left(p + \frac{1}{g}\right) \nonumber = \left(p - \frac{1}{g}\right)\left(1 - p - \frac{1}{g}\right)
\end{align*}
Substituting back:
\begin{equation}
\mathrm{Var}[\hat{f}_v] = \frac{\frac{1}{g}\left(1-\frac{1}{g}\right)}{n\left(p-\frac{1}{g}\right)^2} + \frac{f_v \cdot \left(p - \frac{1}{g}\right)\left(1 - p - \frac{1}{g}\right)}{n \cdot \left(p-\frac{1}{g}\right)^2}
\end{equation}
\begin{equation}
= \frac{\frac{1}{g}\left(1-\frac{1}{g}\right)}{n\left(p-\frac{1}{g}\right)^2} + \frac{f_v\left(1 - p - \frac{1}{g}\right)}{n\left(p-\frac{1}{g}\right)}
\end{equation}
\end{proof}

To verify the correctness of Theorem \ref{thm:variance}, we substitute the optimal OLH parameters $g = e^{\varepsilon} + 1$ and $p = \frac{e^{\varepsilon}}{e^{\varepsilon} + g - 1}$ into Eq.~\ref{eq:variance}. Under this parameterization, $\frac{1}{g} = \frac{1}{e^{\varepsilon}+1}$, and $p - \frac{1}{g} = \frac{1}{2} - \frac{1}{e^{\varepsilon}+1} = \frac{e^{\varepsilon}-1}{2(e^{\varepsilon}+1)}$. Plugging these into Eq.~\ref{eq:variance} and simplifying yields:
\begin{equation}
\mathrm{Var}[\hat{f}_v] = \frac{4e^{\varepsilon}}{n(e^{\varepsilon}-1)^2} + \frac{f_v}{n}
\end{equation}
This matches the known variance of the OLH protocol derived in~\cite{wang2017locally}, confirming that Algorithm \ref{alg:2binom} preserves the variance of the original OLH simulation exactly.

\subsection{Alternative Approach: 3-Binom}

Recall that in Section \ref{sec:insight}, we divided Group 2 into two sub-cases, and then merged them into a single $Y \sim \mathrm{Binomial}\!\left(n - n_v,\, \frac{1}{g}\right)$ at the end. In this section, we present an alternative approach that keeps the two sub-cases separate. Thus, instead of two Binomial random variables, we have three Binomial random variables (hence the name 3-Binom). We show that this alternative formulation also produces unbiased frequency estimates and have identical variance to 2-Binom.

Recall that sub-case A consists of users whose true value is not equal to $v$, but they have a hash collision. Let $n_c$ denote the number of users who have such a hash collision. A key observation is that $n_c$ is not a fixed quantity, but rather, a random variable. Taking into account the collision probability $\frac{1}{g}$, we have: $n_c \sim \mathrm{Binomial}\!\left(n - n_v,\, \frac{1}{g}\right)$. Given the value of $n_c$, the $n - n_v$ Group 2 users are split into two sub-cases. Sub-case A consists of $n_c$ users who experienced a hash collision; each such user contributes to $\mathrm{Sup}(v)$ if and only if the perturbation keeps their hashed value, which occurs with probability $p$. Sub-case B consists of the remaining $n - n_v - n_c$ users who did not experience a hash collision; each such user contributes to $\mathrm{Sup}(v)$ if and only if the perturbation flips their output to the specific value $H_{u_i}(v) \bmod g$, which occurs with probability $q$. Then, in 3-Binom, summing the contributions of each group yields: $\mathrm{Sup}(v) = X + Y_1 + Y_2$, where: $X \sim \mathrm{Binomial}(n_v, p)$ represents Group 1 users (same as in 2-Binom), $Y_1 \sim \mathrm{Binomial}(n_c,\, p)$ represents sub-case A of Group 2, and $Y_2 \sim \mathrm{Binomial}(n - n_v - n_c,\, q)$ represents sub-case B of Group 2.

Algorithm \ref{alg:3binom} shows the simulation algorithm for 3-Binom. It follows the same structure as Algorithm \ref{alg:2binom}; the main difference is between lines 8-11. Here, to simulate the two sub-cases in Group 2, we first draw $n_c$, then use it to draw samples for $y_1$ and $y_2$, and finally, $\mathrm{Sup}(v)$ is the summation of 3 random samples. The time complexity is $O(n+d)$. 

\begin{algorithm}[!t]
\caption{Proposed 3-Binom Simulation Algorithm} \label{alg:3binom}
\begin{algorithmic}[1]
    \Statex \textbf{Input:} Domain $\mathcal{D}$, budget $\varepsilon$, values of users in $\mathcal{P}$ $(v_{u_1}, v_{u_2}, \ldots, v_{u_n})$
    \Statex \textbf{Output:} Estimated frequencies $\hat{f}_v$ for all $v \in \mathcal{D}$
    \State Compute $g \leftarrow \lfloor e^{\varepsilon} + 1 \rfloor$,~~ $p \leftarrow \frac{e^{\varepsilon}}{e^{\varepsilon} + g - 1}$,~~ $q \leftarrow \frac{1}{e^{\varepsilon} + g - 1}$
    \LineComment{Find $n_v$ and $n$ from user values}
    \For{each user $u_i \in \mathcal{P}$}
        \State $n_{v_{u_i}} \leftarrow n_{v_{u_i}} + 1$, \quad $n \leftarrow n + 1$
    \EndFor
    \LineComment{3-Binom approach}
    \For{each value $v \in \mathcal{D}$}
        \State $x \leftarrow$ draw a random sample from $\mathrm{Binomial}(n_v,\, p)$
        \State $n_c \leftarrow$ draw a random sample from $\mathrm{Binomial}\!\left(n - n_v,\, \frac{1}{g}\right)$
        \State $y_1 \leftarrow$ draw a random sample from $\mathrm{Binomial}(n_c,\, p)$
        \State $y_2 \leftarrow$ draw a random sample from $\mathrm{Binomial}(n - n_v - n_c,\, q)$
        \State $\mathrm{Sup}(v) \leftarrow x + y_1 + y_2$
        \State $\hat{f}_v \leftarrow \frac{\mathrm{Sup}(v) - (n/g)}{n \cdot (p - (1/g))}$
    \EndFor
    \State \Return $\hat{f}_v$ for all $v \in \mathcal{D}$
\end{algorithmic}
\end{algorithm}

\subsection{Utility Analysis of 3-Binom}
\label{sec:utility3binom}

We perform a utility analysis for 3-Binom which is similar to 2-Binom. In contrast with 2-Binom, the proofs for 3-Binom require careful treatment of the conditional randomness introduced by $n_c$, and rely on the law of total expectation to handle the two layers of randomness.

\begin{theorem}[Unbiasedness]
The frequency estimates $\hat{f}_v$ produced by Algorithm \ref{alg:3binom} are unbiased, i.e., $\mathbb{E}[\hat{f}_v] = f_v$ for all $v \in \mathcal{D}$.
\end{theorem}

\begin{proof}
Applying the expectation to the estimator on line 12 of Algorithm \ref{alg:3binom}:
\begin{equation} \label{eq:24}
    \mathbb{E}[\hat{f}_v] = \mathbb{E}\!\left[\frac{\mathrm{Sup}(v) - n/g}{n \cdot (p - 1/g)}\right] = \frac{\mathbb{E}[\mathrm{Sup}(v)] - n/g}{n \cdot (p - 1/g)}
\end{equation}
From 3-Binom, we know that: $\mathrm{Sup}(v) = X + Y_1 + Y_2$. By linearity of expectation:
\begin{equation}
    \mathbb{E}[\mathrm{Sup}(v)] = \mathbb{E}[X] + \mathbb{E}[Y_1 + Y_2]
\end{equation}
Since $X \sim \mathrm{Binomial}(n_v, p)$ has fixed parameters, its expectation follows directly: $ \mathbb{E}[X] = n_v \cdot p$. For $\mathbb{E}[Y_1 + Y_2]$, the situation is more involved because $Y_1$ and $Y_2$ have parameters that depend on $n_c$, which is itself a random variable. We therefore apply the law of total expectation, conditioning on $n_c$:
\begin{equation}
    \mathbb{E}[Y_1 + Y_2] = \mathbb{E}_{n_c}\!\left[\mathbb{E}[Y_1 + Y_2 \mid n_c]\right]
\end{equation}
We compute the inner expectation, treating $n_c$ as fixed. By linearity of expectation and the means of the two conditional Binomials:
\begin{align}
    \mathbb{E}[Y_1 + Y_2 \mid n_c] &= \mathbb{E}[Y_1 \mid n_c] + \mathbb{E}[Y_2 \mid n_c] = n_c \cdot p + (n - n_v - n_c) \cdot q \\
    &= n_c \cdot p + (n - n_v) \cdot q - n_c \cdot q = n_c(p - q) + (n - n_v) \cdot q \label{eq:two-term-over-nc}
\end{align}
We take the outer expectation over $n_c$. Since $(n - n_v) \cdot q$ contains no randomness, it passes through the expectation unchanged. Then, we substitute $\mathbb{E}[n_c] = \frac{n - n_v}{g}$, which follows from $n_c \sim \mathrm{Binomial}(n - n_v, \frac{1}{g})$:
\begin{align}
    \mathbb{E}[Y_1 + Y_2] &= \mathbb{E}[n_c](p - q) + (n - n_v) \cdot q \\
    &= \frac{n - n_v}{g}(p - q) + (n - n_v) \cdot q = (n - n_v)\!\left(\frac{p - q}{g} + q\right)
\end{align}
After substituting $p = \frac{e^\varepsilon}{e^\varepsilon + g - 1}$ and $q = \frac{1}{e^\varepsilon + g - 1}$, we find: $\left(\frac{p - q}{g} + q\right) = \frac{1}{g}$ for the rightmost term. Substituting back, we get: $\mathbb{E}[Y_1 + Y_2] = (n - n_v) \cdot \frac{1}{g}$. Combining with $\mathbb{E}[X] = n_v \cdot p$, we obtain:
\begin{equation}
    \mathbb{E}[\mathrm{Sup}(v)] = n_v \cdot p + (n - n_v) \cdot \frac{1}{g}
\end{equation}
which is identical to $\mathbb{E}[\mathrm{Sup}(v)]$ in 2-Binom. The rest of the proof is exactly the same as in Theorem \ref{thm:unbiased}, i.e., substitute $\mathbb{E}[\mathrm{Sup}(v)]$ into the estimator (Eq.~\ref{eq:24}) and simplify to arrive at $\mathbb{E}[\hat{f}_v] = f_v$.
\end{proof}

\begin{theorem}[Variance]
\label{thm:variance-three}
The variance of the frequency estimates $\hat{f}_v$ produced by Algorithm \ref{alg:3binom} is identical to that of Algorithm \ref{alg:2binom}, given by:
\begin{equation}
\mathrm{Var}[\hat{f}_v] = \frac{\frac{1}{g}\left(1 - \frac{1}{g}\right)}{n\left(p - \frac{1}{g}\right)^2} + \frac{f_v\left(1 - p - \frac{1}{g}\right)}{n\left(p - \frac{1}{g}\right)}
\end{equation}
\end{theorem}

\begin{proof}
We compute the variance of the estimator on line 12 of Algorithm \ref{alg:3binom}. Since $n$ and $g$ are constants:
\begin{equation}
    \mathrm{Var}[\hat{f}_v] = \frac{\mathrm{Var}[\mathrm{Sup}(v)]}{n^2 \cdot (p - 1/g)^2}
\end{equation}
Since $X$ is independent of $Y_1$ and $Y_2$, the variance of $\mathrm{Sup}(v) = X + Y_1 + Y_2$ separates as: $\mathrm{Var}[\mathrm{Sup}(v)] = \mathrm{Var}[X] + \mathrm{Var}[Y_1 + Y_2]$. The first term follows directly from $X \sim \mathrm{Binomial}(n_v, p)$: $\mathrm{Var}[X] = n_v \cdot p \cdot (1-p)$. For the second term, since $Y_1$ and $Y_2$ share the random parameter $n_c$, they are not independent of each other and their variance cannot be computed by simple summation. We instead apply the law of total variance, conditioning on $n_c$:
\begin{equation}
    \mathrm{Var}[Y_1 + Y_2] = \mathbb{E}_{n_c}[\mathrm{Var}[Y_1 + Y_2 \mid n_c]] + \mathrm{Var}_{n_c}[\mathbb{E}[Y_1 + Y_2 \mid n_c]]
\end{equation}
For the first component, $Y_1$ and $Y_2$ are independent conditional on $n_c$, so their variances sum directly:
\begin{equation}
    \mathrm{Var}[Y_1 + Y_2 \mid n_c] = n_c \cdot p(1-p) + (n - n_v - n_c) \cdot q(1-q)
\end{equation}
Taking the expectation over $n_c$ and substituting $\mathbb{E}[n_c] = \frac{n-n_v}{g}$:
\begin{equation}
    \mathbb{E}_{n_c}[\mathrm{Var}[Y_1 + Y_2 \mid n_c]] = \frac{n-n_v}{g}\left[p(1-p) + (g-1)q(1-q)\right]
\end{equation}
For the second component, from the unbiasedness proof we have $\mathbb{E}[Y_1 + Y_2 \mid n_c] = n_c \cdot (p-q) + (n-n_v) \cdot q$ as can be found in Eq.~\ref{eq:two-term-over-nc}. Since $(n-n_v) \cdot q$ is a constant with respect to $n_c$:
\begin{equation}
    \mathrm{Var}_{n_c}[\mathbb{E}[Y_1 + Y_2 \mid n_c]] = (p-q)^2 \cdot \mathrm{Var}[n_c] = (p-q)^2 \cdot (n-n_v) \cdot \frac{1}{g}\left(1 - \frac{1}{g}\right)
\end{equation}
Combining both components:
\begin{equation}
    \mathrm{Var}[Y_1 + Y_2] = \frac{n-n_v}{g}\left[p(1-p) + (g-1)q(1-q) + (p-q)^2\left(1-\frac{1}{g}\right)\right]
\end{equation}
Substituting $p = \frac{e^\varepsilon}{e^\varepsilon+g-1}$ and $q = \frac{1}{e^\varepsilon+g-1}$ and simplifying the bracket (see Appendix~\ref{appendix:bracket}) yields:
\begin{equation}
    \mathrm{Var}[Y_1 + Y_2] = (n - n_v) \cdot \frac{1}{g} \cdot \left(1 - \frac{1}{g}\right)
\end{equation}
This is identical to $\mathrm{Var}[Y]$ in the 2-Binom proof (see Eq.~\ref{eq:emre}). Thus, $\mathrm{Var}[\mathrm{Sup}(v)]$ is identical in 2-Binom and 3-Binom. Consequently, the remainder of the proof is exactly as in Theorem \ref{thm:variance}, yielding the same variance result.
\end{proof}
 
\section{Experimental Evaluation}

In this section, we present the experimental evaluation of our Binomial-based simulation methodologies (2-Binom and 3-Binom) by comparing them against existing OLH simulations (Algorithm \ref{algo:olh}). Our evaluation has two main goals: (i) to verify that the outputs of 2-Binom and 3-Binom are statistically indistinguishable from those of existing OLH simulations, and (ii) to demonstrate that our methods achieve substantial execution time improvements.

\subsection{Experiment Setup}

All algorithms were implemented in Python. Simulations were performed on an Intel Core i9-14900K CPU with 64 GB RAM, running Ubuntu 22.04.5 LTS. We used four real-world datasets that are frequently used in the LDP literature:
 
\textbf{Adult} contains individuals' census-related information\footnote{\url{https://archive.ics.uci.edu/dataset/2/adult}}. We used the age attribute as users' values, which range from 17 to 90, yielding a domain of size $d = 74$. The dataset contains $n = 45{,}222$ records.

\textbf{MSNBC} consists of browsing activity logs of users who visited \url{msnbc.com} on September 28, 1999\footnote{\url{https://archive.ics.uci.edu/dataset/133/msnbc+com+anonymous+web+data}}. Each record corresponds to a sequence of page visits for a user, where each visit is represented by a content category (e.g., news, technology, weather, sports). The domain size is $d = 17$ categories. For each user, we consider only the first entry in their sequence (i.e., the first visited category) as their $v_u$. The population size is $n = 989{,}818$.

\textbf{Kosarak} is a clickstream dataset derived from a Hungarian online news portal\footnote{\url{https://www.philippe-fournier-viger.com/spmf/index.php?link=datasets.php}}. Each record represents the sequence of URLs visited by a user. Due to the sparsity of many URLs, we restricted the domain to $d = 128$ most frequently visited URLs. For users with multiple URLs, the URL with the highest frequency in their sequence was selected as their $v_u$. 

\textbf{BMS-POS} comprises market basket transaction data collected from a large electronics retailer, containing $n = 500{,}000$ transactions and 1,657 distinct items\footnote{\url{https://github.com/cpearce/HARM/blob/master/datasets/BMS-POS.csv}}. We applied the same preprocessing procedure as in Kosarak, retaining the $d = 256$ most frequently purchased items.

Simulations were conducted for privacy budget values $\varepsilon \in \{0.5,\, 1.0,\, 1.5,\, 2.0,\,$ $2.5,\, 3.0,\, 3.5,\, 4.0\}$. Execution time experiments are repeated 10 times, and the estimation error and variance experiments are repeated 5,000 times to ensure statistical significance.
 
\subsection{Execution Time Experiments}

\textbf{OLH versus 2-Binom and 3-Binom.} We evaluate the execution time of OLH versus 2-Binom and 3-Binom across all four datasets and all $\varepsilon$ values. For each configuration, we measure the total time over 10 simulation repetitions. The results are reported in
Table~\ref{tab:exec_time}. It can be observed that 2-Binom and 3-Binom achieve execution times multiple orders of magnitude lower than OLH, completing every configuration in under 5 milliseconds, while OLH ranges from seconds to several minutes. For example, on the MSNBC dataset, OLH requires between 54 and 63 seconds, whereas 2-Binom and 3-Binom complete the same task in under 0.5 milliseconds. On larger datasets such as BMS-POS and Kosarak, OLH times exceed 3-4 minutes, whereas both Binomial methods still finish in well under a second.

\begin{table}[ht]
\centering
\caption{Total execution times (in seconds) for 10 repetitions across different $\varepsilon$ values.}
\label{tab:exec_time}
\setlength{\tabcolsep}{3pt}
\begin{tabular}{llcccccccc}
\toprule
\textbf{Dataset} & \textbf{Method}
  & $\varepsilon{=}0.5$ & $\varepsilon{=}1.0$ & $\varepsilon{=}1.5$
  & $\varepsilon{=}2.0$ & $\varepsilon{=}2.5$ & $\varepsilon{=}3.0$
  & $\varepsilon{=}3.5$ & $\varepsilon{=}4.0$ \\
\midrule
\multirow{3}{*}{Adult}
  & OLH            & 8.6496 & 8.3656 & 8.1026 & 7.8320 & 7.7127 & 7.6063 & 7.4373 & 7.3598 \\
  & 2-Binom       & 0.0008 & 0.0009 & 0.0008 & 0.0008 & 0.0008 & 0.0008 & 0.0008 & 0.0008 \\
  & 3-Binom  & 0.0012 & 0.0014 & 0.0013 & 0.0013 & 0.0013 & 0.0013 & 0.0012 & 0.0013 \\
\midrule
\multirow{3}{*}{MSNBC}
  & OLH            & 63.037 & 60.954 & 59.021 & 57.635 & 56.345 & 55.832 & 54.491 & 54.867 \\
  & 2-Binom       & 0.0004 & 0.0004 & 0.0004 & 0.0004 & 0.0004 & 0.0004 & 0.0004 & 0.0004 \\
  & 3-Binom  & 0.0004 & 0.0004 & 0.0004 & 0.0004 & 0.0004 & 0.0004 & 0.0004 & 0.0004 \\
\midrule
\multirow{3}{*}{Kosarak}
  & OLH            & 159.75 & 151.31 & 147.21 & 141.59 & 138.33 & 135.33 & 133.87 & 133.36 \\
  & 2-Binom       & 0.0015 & 0.0015 & 0.0015 & 0.0015 & 0.0015 & 0.0015 & 0.0015 & 0.0015 \\
  & 3-Binom  & 0.0018 & 0.0019 & 0.0019 & 0.0019 & 0.0019 & 0.0019 & 0.0018 & 0.0019 \\
\midrule
\multirow{3}{*}{BMS-POS}
  & OLH            & 304.37 & 289.84 & 284.95 & 274.70 & 265.88 & 262.16 & 257.32 & 252.03 \\
  & 2-Binom       & 0.0025 & 0.0025 & 0.0029 & 0.0025 & 0.0025 & 0.0027 & 0.0025 & 0.0025 \\
  & 3-Binom  & 0.0036 & 0.0036 & 0.0041 & 0.0036 & 0.0036 & 0.0036 & 0.0036 & 0.0036 \\
\bottomrule
\end{tabular}
\end{table}

The results are consistent with the algorithms' computational complexities. OLH exhibits an $O(nd)$ cost, causing execution time to increase with both the number of users and the domain size. In particular, datasets with large domains, such as BMS-POS and Kosarak, incur substantially higher runtimes than MSNBC despite MSNBC having the largest population. In contrast, both Binomial methods have lower complexity. It should be noted that 3-Binom incurs a small additional time overhead compared to 2-Binom due to the extra Binomial draws per $v \in \mathcal{D}$, but the effect of this is very small. Overall, these results demonstrate the substantial efficiency advantage of 2-Binom and 3-Binom over existing OLH simulations.

\textbf{Impact of parallelizing OLH.} The previous experiments used the standard single-threaded OLH simulations as the baseline. To understand whether the gap between OLH and our methods can be closed by parallelizing the OLH simulations, we implemented a multi-threaded version of Algorithm \ref{alg:olh} in which the user population is partitioned into $\eta$ disjoint chunks, and each chunk is assigned to a different CPU thread. Each thread independently computes $\mathit{Sup}(v)$ and $\hat{f}_v$ for its population subset, and the results are combined at the end. We report execution times with $\eta \in \{1, 2, 4, 8, 16\}$ and $\varepsilon = 1$ over 10 repetitions in Table \ref{tab:parallel}, where 2-Binom and 3-Binom remain single-threaded.

We observe that parallelization reduces OLH's execution time, but its benefit diminishes as more threads are added, and the reduction is not sufficient to close the gap between OLH versus 2-Binom and 3-Binom. Even at $\eta = 16$, OLH still requires several seconds to tens of seconds on every dataset, whereas single-threaded 2-Binom and 3-Binom complete in well under a millisecond throughout. Our methods therefore remain several orders of magnitude faster than OLH at every thread count we tested.

\begin{table}[!t]
\centering
\caption{Total execution times (in seconds) for 10 repetitions at $\varepsilon = 1.0$, comparing 2-Binom and 3-Binom against parallel OLH implementations across a varying number of threads $\eta$. The 2-Binom and 3-Binom methods are single-threaded.}
\label{tab:parallel}
\setlength{\tabcolsep}{5pt}
\begin{tabular}{lrrrrrcc}
\toprule
& \multicolumn{5}{c}{\textbf{OLH (parallel, $\eta$ threads)}} & \multicolumn{2}{c}{\textbf{Binomial}} \\
\cmidrule(lr){2-6} \cmidrule(lr){7-8}
\textbf{Dataset}
  & $\eta$ = 1 & $\eta$ = 2 & $\eta$ = 4 & $\eta$ = 8 & $\eta$ = 16
  & \textbf{2-Binom} & \textbf{3-Binom} \\
\midrule
Adult   & 8.5531 & 4.3803 & 2.6380 & 1.7852 & 1.4543 & 0.0005 & 0.0010 \\
MSNBC   & 62.739 & 34.907 & 20.468 & 14.625 & 12.581 & 0.0002 & 0.0003 \\
Kosarak & 155.14 & 78.902 & 48.702 & 32.588 & 25.403 & 0.0009 & 0.0017 \\
BMS-POS & 304.15 & 154.50 & 91.823 & 65.495 & 52.596 & 0.0018 & 0.0034 \\
\bottomrule
\end{tabular}
\end{table}

\textbf{Variability of execution times.} To assess run-to-run variability and verify the stability of the reported speedups, we report the standard deviations of execution times (in addition to the means) across 10 experiment repetitions. Table \ref{tab:variability} reports the results for Adult and BMS-POS, which are the two datasets with smallest and largest execution times according to our previous experiments. MSNBC and Kosarak datasets exhibit similar standard deviation behaviors and are omitted due to space.

We observe that OLH exhibits low standard deviations across all configurations. This stability is intuitive because each OLH run lasts several seconds, which reduces the impact of transient system effects. In contrast, 2-Binom and 3-Binom exhibit relatively larger standard deviations compared to their means. Our inspection of individual runs showed that this is caused by a small number of isolated slow executions; and we note that in every configuration, the median run time remains within a few percent of the mean. It should also be noted that microsecond-scale system-related impacts can be more pronounced in 2-Binom and 3-Binom, since their execution times are on the order of microseconds. More importantly, considering that the OLH results are reported in seconds whereas 2-Binom and 3-Binom are reported in $\mu$s, the observed standard deviations are far too small to account for the reported gains. The speedups between OLH and the Binomial methods span several orders of magnitude, whereas the standard deviations remain at most a few percent of the mean. This supports the conclusion that the improvements arise from the proposed algorithms rather than from measurement artifacts.


\begin{table}[!t]
\centering
\small
\caption{Execution time per simulation repetition on Adult and BMS-POS, reported as mean $\pm$ standard deviation. OLH times are in seconds; 2-Binom and 3-Binom times are in microseconds.}
\label{tab:timing_variability}
\setlength{\tabcolsep}{3pt}
\begin{tabular}{lcccccc}
\toprule
& \multicolumn{3}{c}{\textbf{Adult}} & \multicolumn{3}{c}{\textbf{BMS-POS}} \\
\cmidrule(lr){2-4} \cmidrule(lr){5-7}
$\boldsymbol{\varepsilon}$
  & \textbf{OLH (s)} & \textbf{2-Binom} & \textbf{3-Binom}
  & \textbf{OLH (s)} & \textbf{2-Binom} & \textbf{3-Binom} \\
  & & \textbf{($\mu$s)} & \textbf{($\mu$s)} & & \textbf{($\mu$s)} & \textbf{($\mu$s)} \\
\midrule
$0.5$ & $0.927 \pm 0.013$ & $56.4 \pm 7.3$ & $104.1 \pm 9.6$ & $31.862 \pm 0.020$ & $179.7 \pm 2.7$ & $334.9 \pm 4.1$ \\
$1.0$ & $0.882 \pm 0.008$ & $54.0 \pm 1.1$ & $100.3 \pm 1.6$ & $30.393 \pm 0.013$ & $179.5 \pm 2.4$ & $334.4 \pm 4.9$ \\
$1.5$ & $0.841 \pm 0.001$ & $53.7 \pm 0.8$ & $100.3 \pm 1.7$ & $29.786 \pm 0.021$ & $187.0 \pm 33.0$ & $338.9 \pm 18.9$ \\
$2.0$ & $0.808 \pm 0.001$ & $55.8 \pm 1.1$ & $101.1 \pm 2.1$ & $28.156 \pm 0.013$ & $176.9 \pm 2.2$ & $333.4 \pm 3.2$ \\
$2.5$ & $0.788 \pm 0.002$ & $54.2 \pm 1.0$ & $101.3 \pm 1.9$ & $27.348 \pm 0.017$ & $178.2 \pm 3.6$ & $336.2 \pm 10.4$ \\
$3.0$ & $0.774 \pm 0.001$ & $54.0 \pm 0.7$ & $100.4 \pm 1.4$ & $27.453 \pm 0.315$ & $182.1 \pm 22.2$ & $347.8 \pm 63.4$ \\
$3.5$ & $0.766 \pm 0.001$ & $54.9 \pm 1.2$ & $101.3 \pm 1.9$ & $26.995 \pm 0.017$ & $177.0 \pm 2.4$ & $334.2 \pm 4.6$ \\
$4.0$ & $0.760 \pm 0.001$ & $54.7 \pm 0.9$ & $102.6 \pm 6.3$ & $26.748 \pm 0.011$ & $178.9 \pm 2.2$ & $339.6 \pm 6.2$ \\
\bottomrule
\end{tabular}
\label{tab:variability}
\end{table}

\subsection{Memory Consumption Experiments}

Beyond execution time, the memory footprint of simulations can also be important. We compare the peak memory of OLH against 2-Binom and 3-Binom across all four datasets and all $\varepsilon$ values, measured using Python's \texttt{tracemalloc} module on a preloaded dataset so that only the simulation phase is captured. The results are given in Table \ref{tab:memory}. Across all datasets and privacy budgets, both Binomial methods consume several orders of magnitude less memory than OLH: every Binomial configuration completes within tens of kilobytes, whereas OLH ranges from a few megabytes on Adult to tens of megabytes on MSNBC. Importantly, the effect of $n$ and $d$ on OLH's memory differs from their effect on its runtime. Algorithm \ref{alg:olh} stores the report $\langle x'_u, H_u \rangle$ of every user in $L$, giving $O(n+d)$ memory which is dominated by $n$. This is visible in the results: Kosarak and BMS-POS consume a near-identical 42.14 MB memory despite their domains differing by a factor of two ($d = 128$ versus $d = 256$), because both contain $n = 500{,}000$ users. On the other hand, MSNBC is the most memory-intensive dataset despite its small $d$, due to its large $n$. 

In contrast, 2-Binom and 3-Binom store only the per-value counts $n_v$ and the support array, both of size $d$, so their memory is $O(d)$ and independent of $n$. Consumption grows only modestly from MSNBC ($d = 17$) to BMS-POS ($d = 256$), even though MSNBC's population is nearly twice as large. Unlike the execution time results, 2-Binom and 3-Binom are indistinguishable in memory, since the extra Binomial draw in 3-Binom does not require much additional storage. All methods are stable across varying $\varepsilon$, confirming that 2-Binom and 3-Binom offer substantial memory benefits compared to OLH, alongside their execution time benefits.

\begin{table}[!t]
\centering
\caption{Memory consumptions (in MB) of simulations across different $\varepsilon$ values.}
\label{tab:mem_consumption}
\setlength{\tabcolsep}{3pt}
\begin{tabular}{llcccccccc}
\toprule
\textbf{Dataset} & \textbf{Method}
  & $\varepsilon{=}0.5$ & $\varepsilon{=}1.0$ & $\varepsilon{=}1.5$
  & $\varepsilon{=}2.0$ & $\varepsilon{=}2.5$ & $\varepsilon{=}3.0$
  & $\varepsilon{=}3.5$ & $\varepsilon{=}4.0$ \\
\midrule
\multirow{3}{*}{Adult}
  & OLH            & 3.8444 & 3.8434 & 3.8434 & 3.8433 & 3.8446 & 3.8435 & 3.8433 & 3.8435 \\
  & 2-Binom       & 0.0157 & 0.0151 & 0.0152 & 0.0158 & 0.0152 & 0.0158 & 0.0150 & 0.0151 \\
  & 3-Binom  & 0.0152 & 0.0157 & 0.0158 & 0.0157 & 0.0157 & 0.0151 & 0.0153 & 0.0157 \\
\midrule
\multirow{3}{*}{MSNBC}
  & OLH            & 83.591 & 83.591 & 83.593 & 83.591 & 83.591 & 83.593 & 83.590 & 83.591 \\
  & 2-Binom       & 0.0150 & 0.0150 & 0.0150 & 0.0152 & 0.0151 & 0.0151 & 0.0153 & 0.0153 \\
  & 3-Binom  & 0.0151 & 0.0150 & 0.0150 & 0.0152 & 0.0152 & 0.0151 & 0.0151 & 0.0153 \\
\midrule
\multirow{3}{*}{Kosarak}
  & OLH            & 42.138 & 42.140 & 42.138 & 42.142 & 42.138 & 42.140 & 42.140 & 42.138 \\
  & 2-Binom       & 0.0161 & 0.0161 & 0.0157 & 0.0157 & 0.0156 & 0.0161 & 0.0161 & 0.0154 \\
  & 3-Binom  & 0.0157 & 0.0160 & 0.0160 & 0.0161 & 0.0161 & 0.0162 & 0.0161 & 0.0161 \\
\midrule
\multirow{3}{*}{BMS-POS}
  & OLH            & 42.143 & 42.139 & 42.143 & 42.139 & 42.139 & 42.139 & 42.143 & 42.139 \\
  & 2-Binom       & 0.0165 & 0.0170 & 0.0171 & 0.0166 & 0.0164 & 0.0170 & 0.0170 & 0.0172 \\
  & 3-Binom  & 0.0166 & 0.0172 & 0.0169 & 0.0171 & 0.0170 & 0.0171 & 0.0165 & 0.0166 \\
\bottomrule
\end{tabular}
\label{tab:memory}
\end{table}

\subsection{Comparison of Estimation Error}

To verify that the outputs of 2-Binom and 3-Binom are statistically indistinguishable from those of existing OLH simulations, we compare their estimation errors across all datasets and privacy budgets. We measure error using $\ell_1$ and $\ell_2$ distances between the estimated frequencies and the true frequencies:
\begin{equation}
    \ell_1 \text{ distance} = \sum_{v \in \mathcal{D}} \lvert f_v - \hat{f}_v \rvert
    \qquad
    \ell_2 \text{ distance} = \sqrt{\sum_{v \in \mathcal{D}} (f_v - \hat{f}_v)^2}
\end{equation}

We report the results in Figure \ref{fig:estimation_error}. It can be observed that the curves for OLH, 2-Binom, and 3-Binom are nearly indistinguishable across all datasets and privacy budgets, indicating that the proposed simulations accurately reproduce the behavior of OLH. For all methods, both $\ell_1$ and $\ell_2$ distances decrease as $\varepsilon$ increases, reflecting the reduction in perturbation noise and the resulting improvement in estimation accuracy. Overall, the proposed methods achieve the same estimation results as OLH while providing the substantial time savings shown in Table \ref{tab:exec_time}.

\begin{figure}[t]
    \centering
    \includegraphics[width=0.24\linewidth]{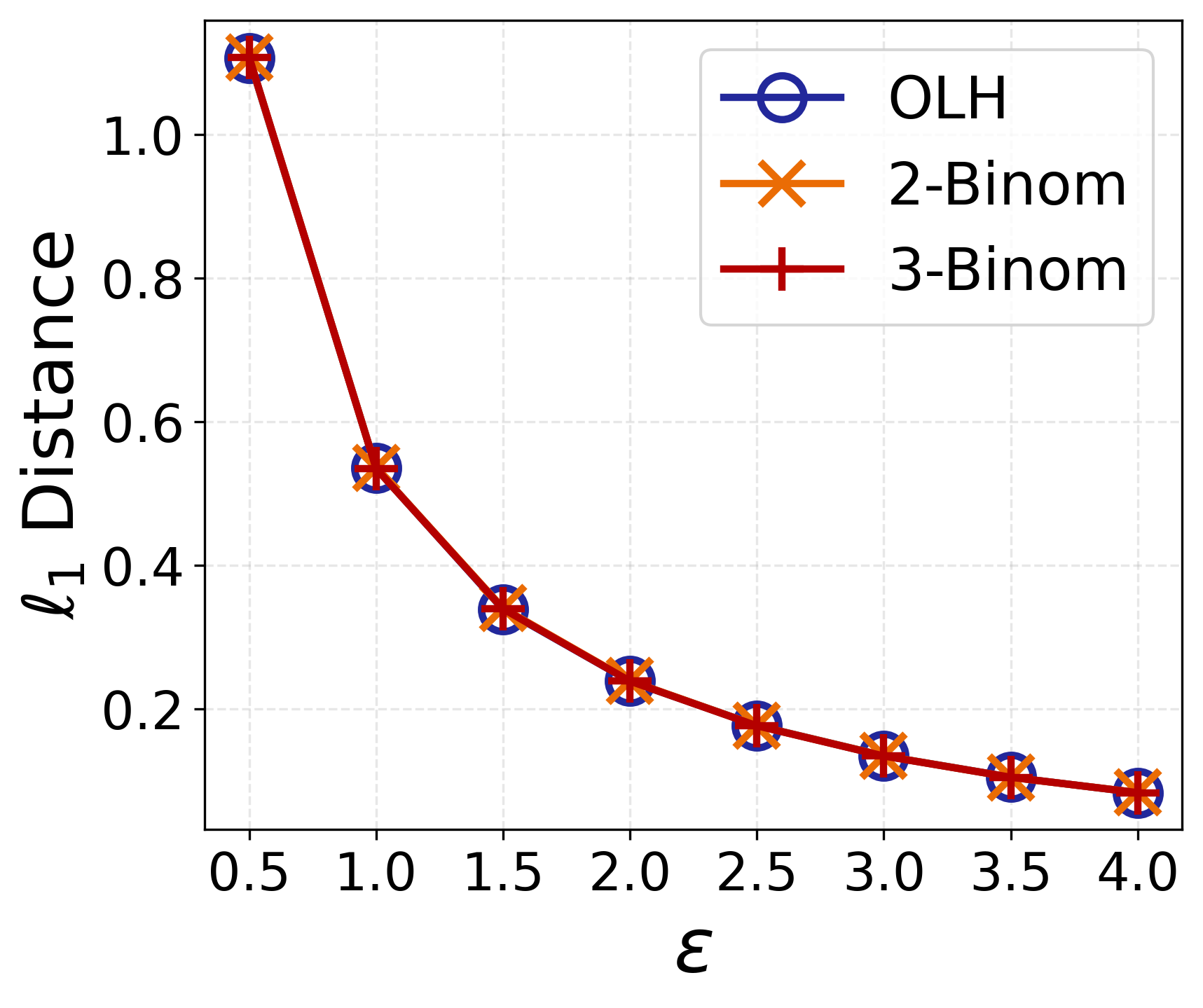}\hfill
    \includegraphics[width=0.25\linewidth]{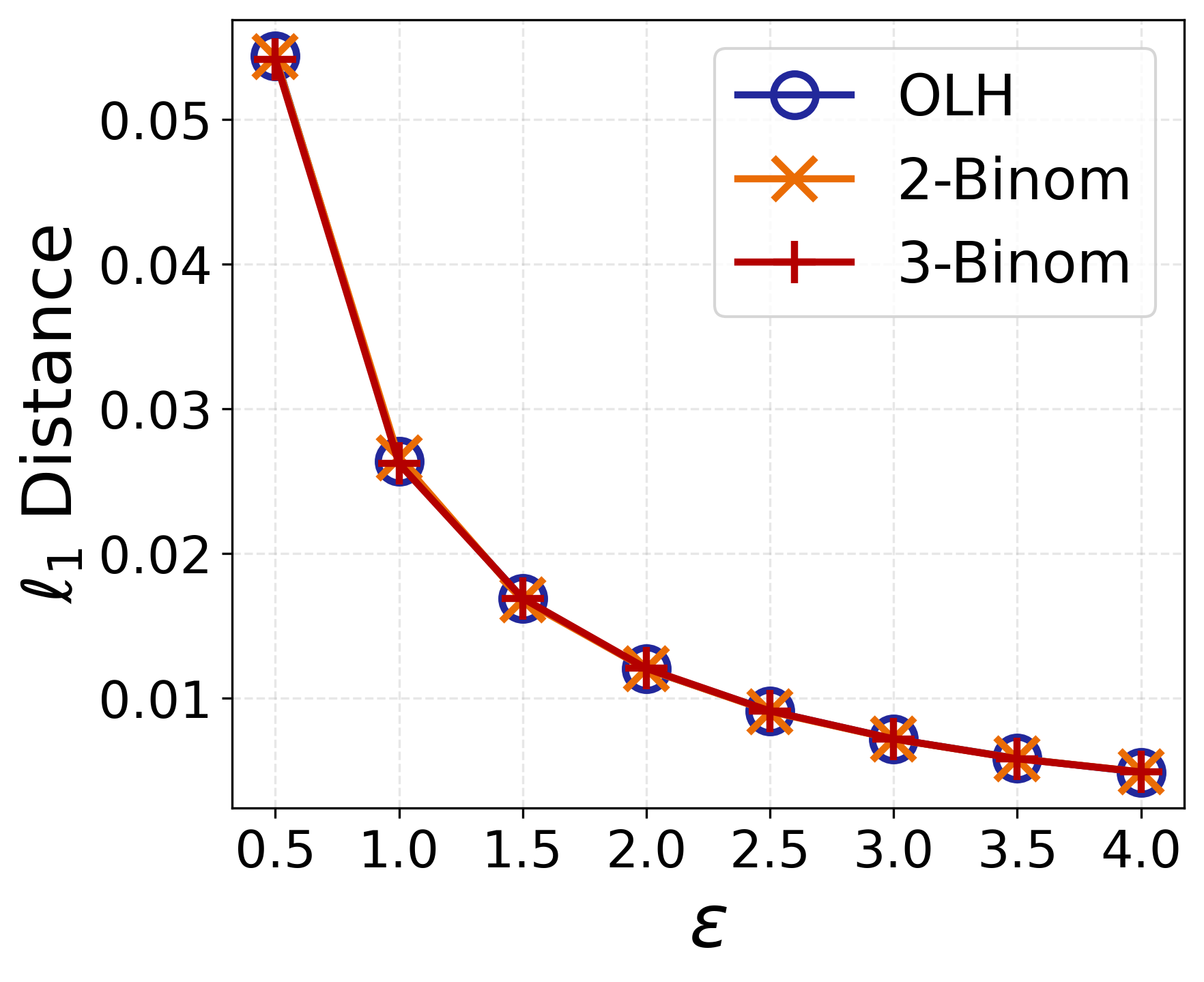}\hfill
    \includegraphics[width=0.24\linewidth]{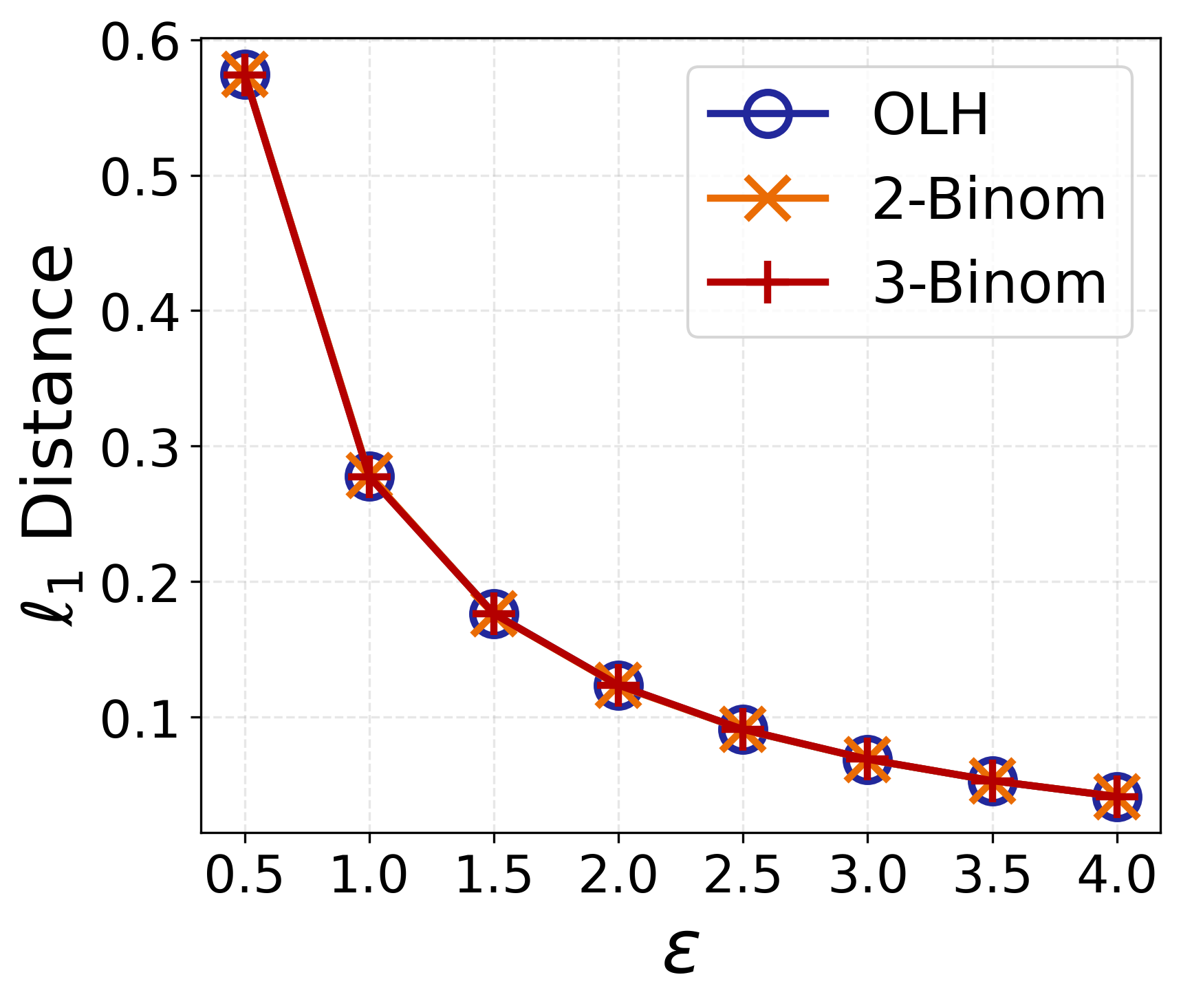}\hfill
    \includegraphics[width=0.24\linewidth]{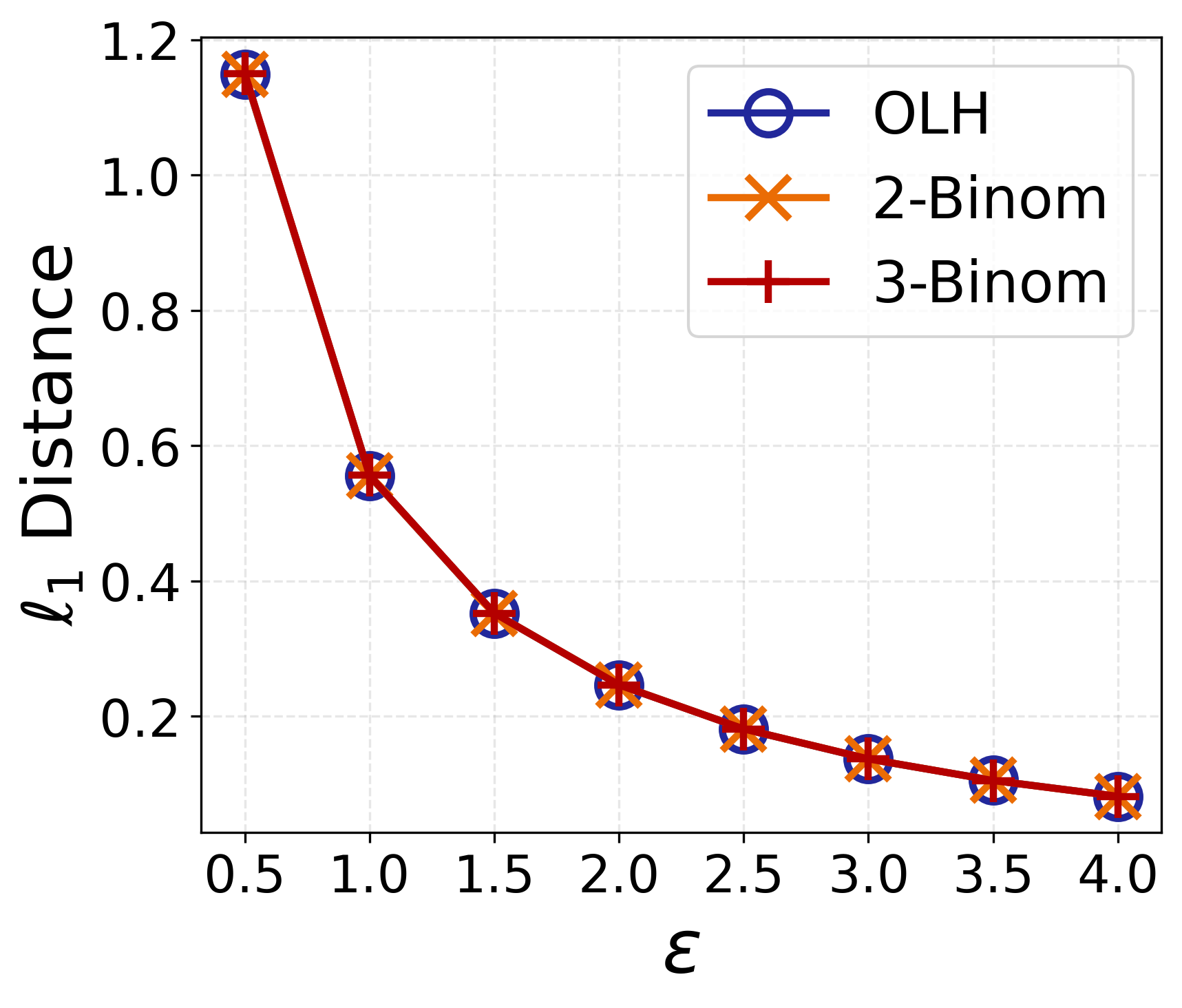}\\[2pt]
    \includegraphics[width=0.24\linewidth]{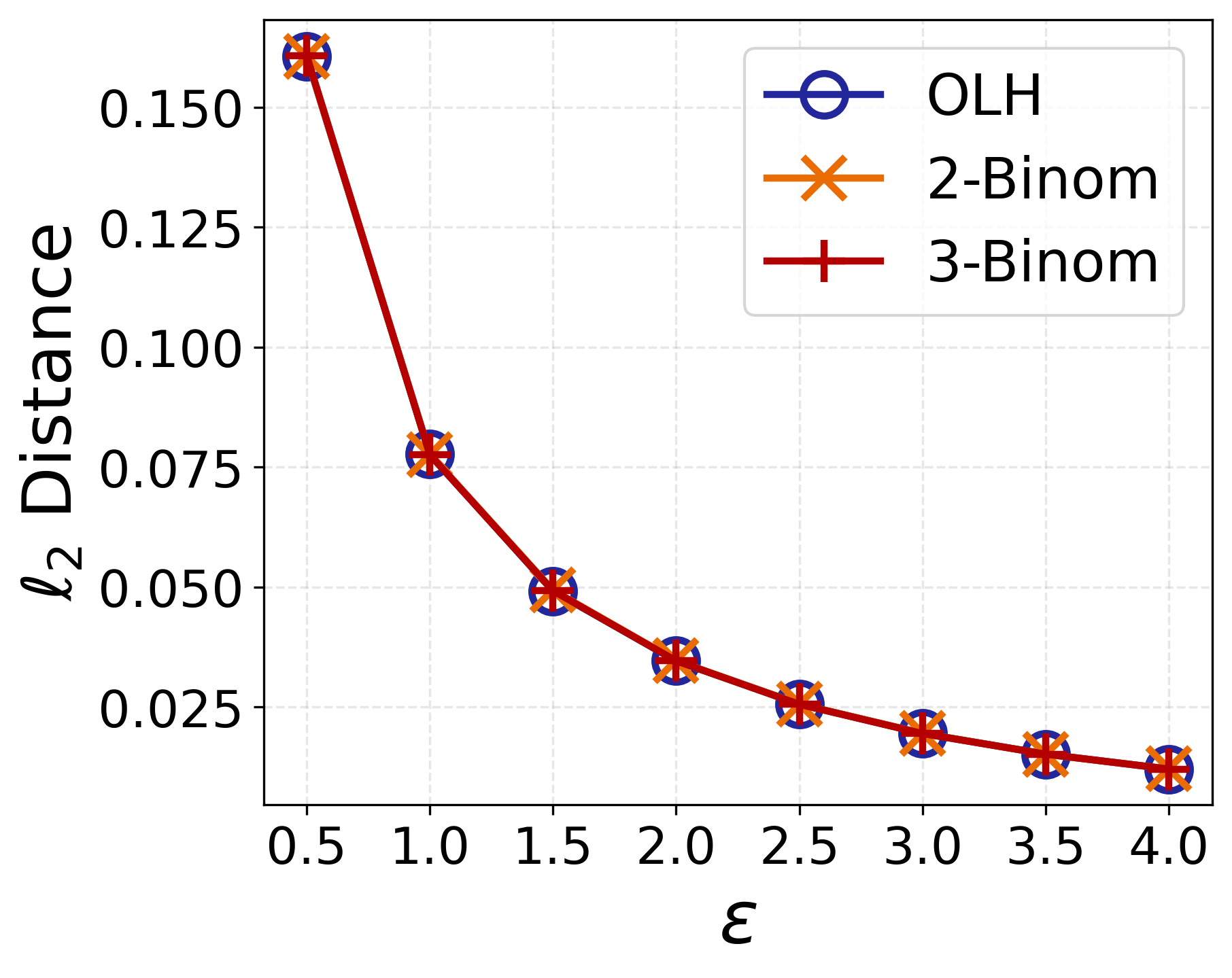}\hfill
    \includegraphics[width=0.25\linewidth]{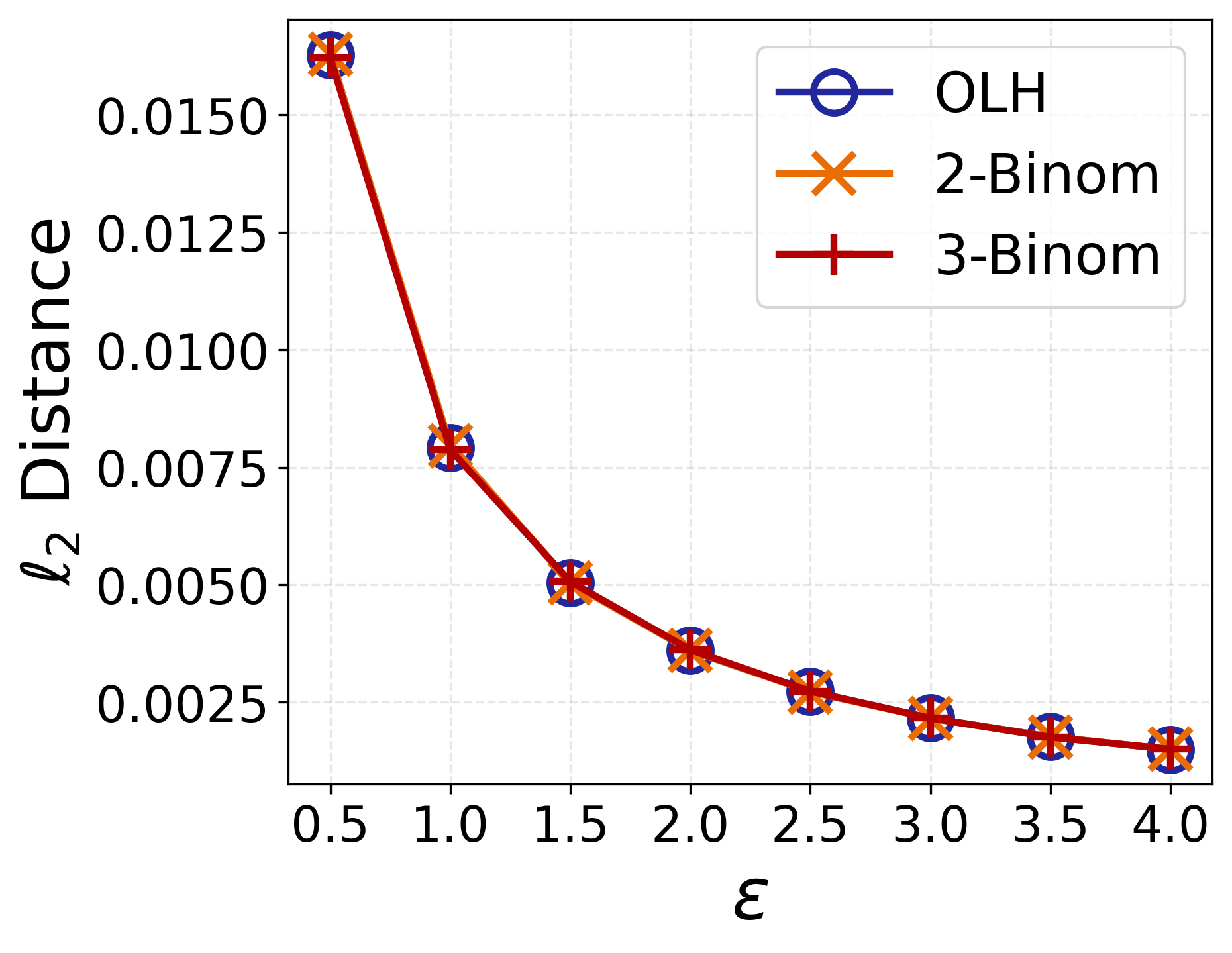}\hfill
    \includegraphics[width=0.24\linewidth]{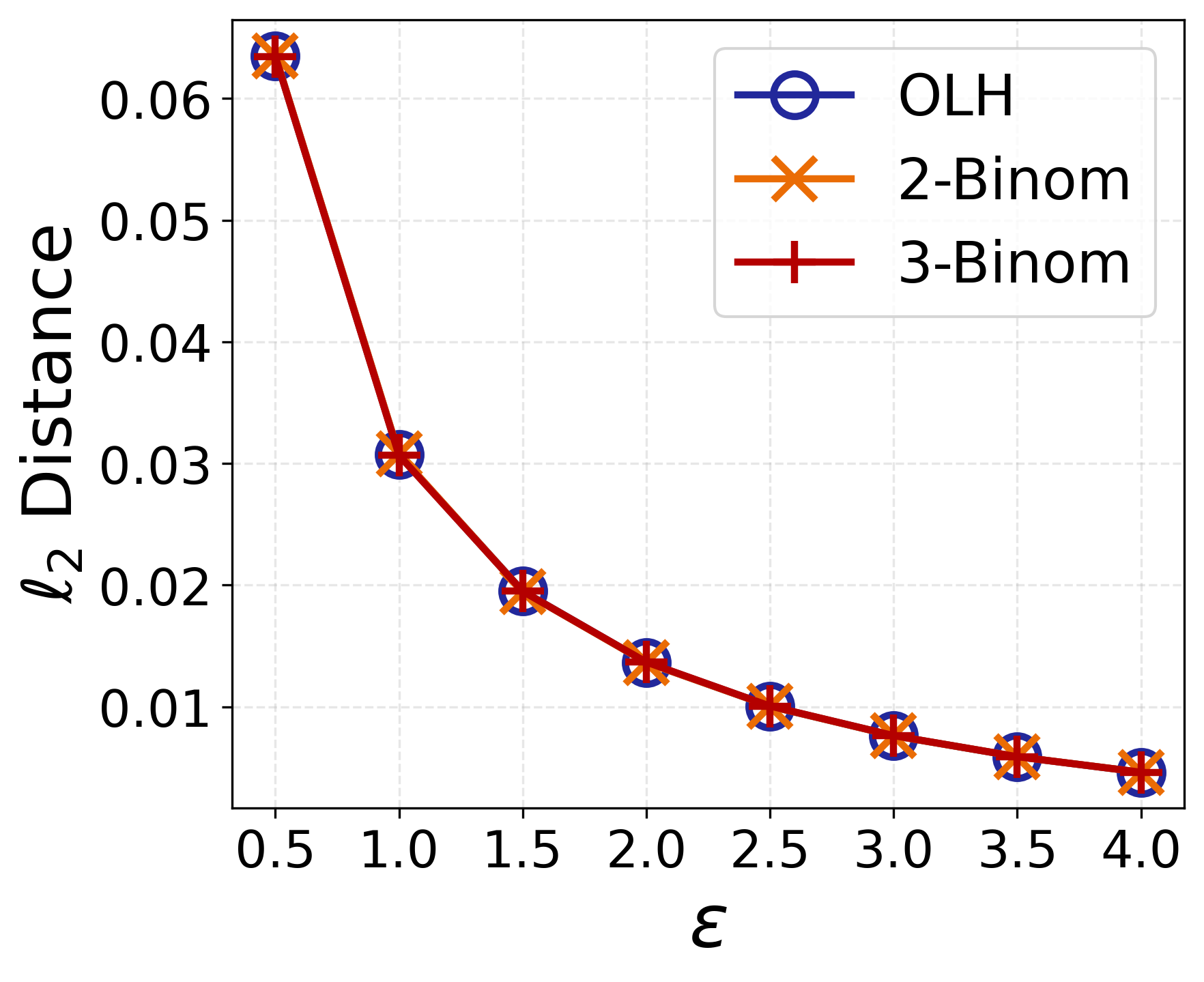}\hfill
    \includegraphics[width=0.24\linewidth]{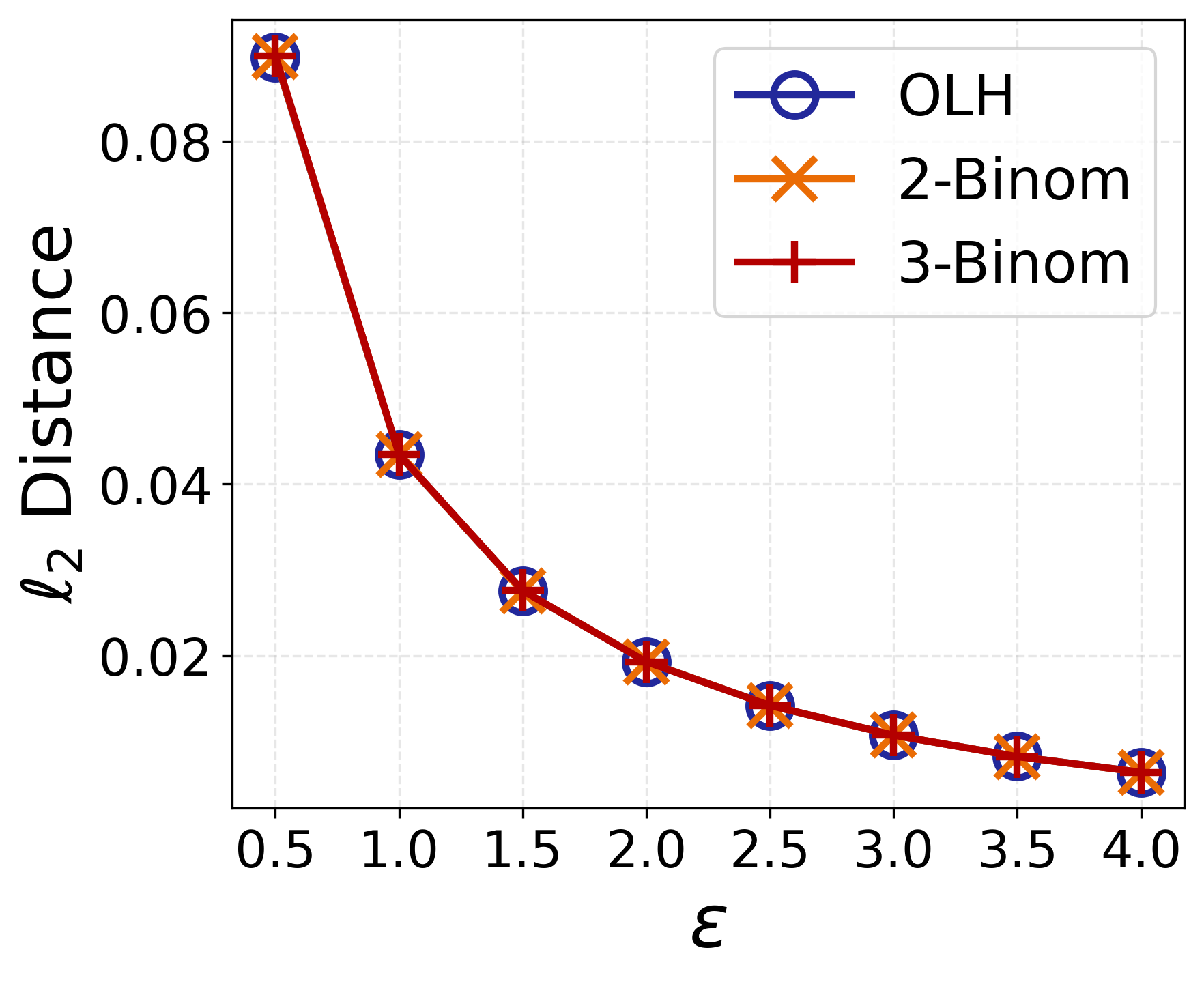}
    \vspace{-8pt}
    \caption{Average $\ell_1$ and $\ell_2$ estimation errors of OLH, 2-Binom, and 3-Binom across varying $\varepsilon$ values. Datasets from left to right: Adult, MSNBC, Kosarak, BMS-POS.}
    \label{fig:estimation_error}
\end{figure}
 
To further examine the distributional similarity of the three methods, we use kernel density estimation (KDE) to visualize the full distribution of $\ell_1$ and $\ell_2$ distances across the 5{,}000 repetitions at $\varepsilon = 2.0$. The resulting density plots are shown in Figure \ref{fig:kde}, with $\ell_1$ distributions in the top row and $\ell_2$ distributions in the bottom row. In all plots, the three density curves for OLH, 2-Binom, and 3-Binom overlap closely. This indicates that the distributions of estimation errors produced by OLH, 2-Binom, and 3-Binom are nearly identical. Furthermore, the KDE plots reveal that the distributions are approximately bell-shaped in all datasets, with the three methods sharing the same center and spread. Any visible separation between the curves is minor and plausible due to the randomized nature of OLH. The tightest agreement is observed on MSNBC ($d = 17$), where the small domain causes each method's per-repetition noise to be highly concentrated, while on BMS-POS and Kosarak the distributions are slightly wider owing to the larger domains. 

\begin{figure}[t]
    \centering
    \includegraphics[width=0.24\linewidth]{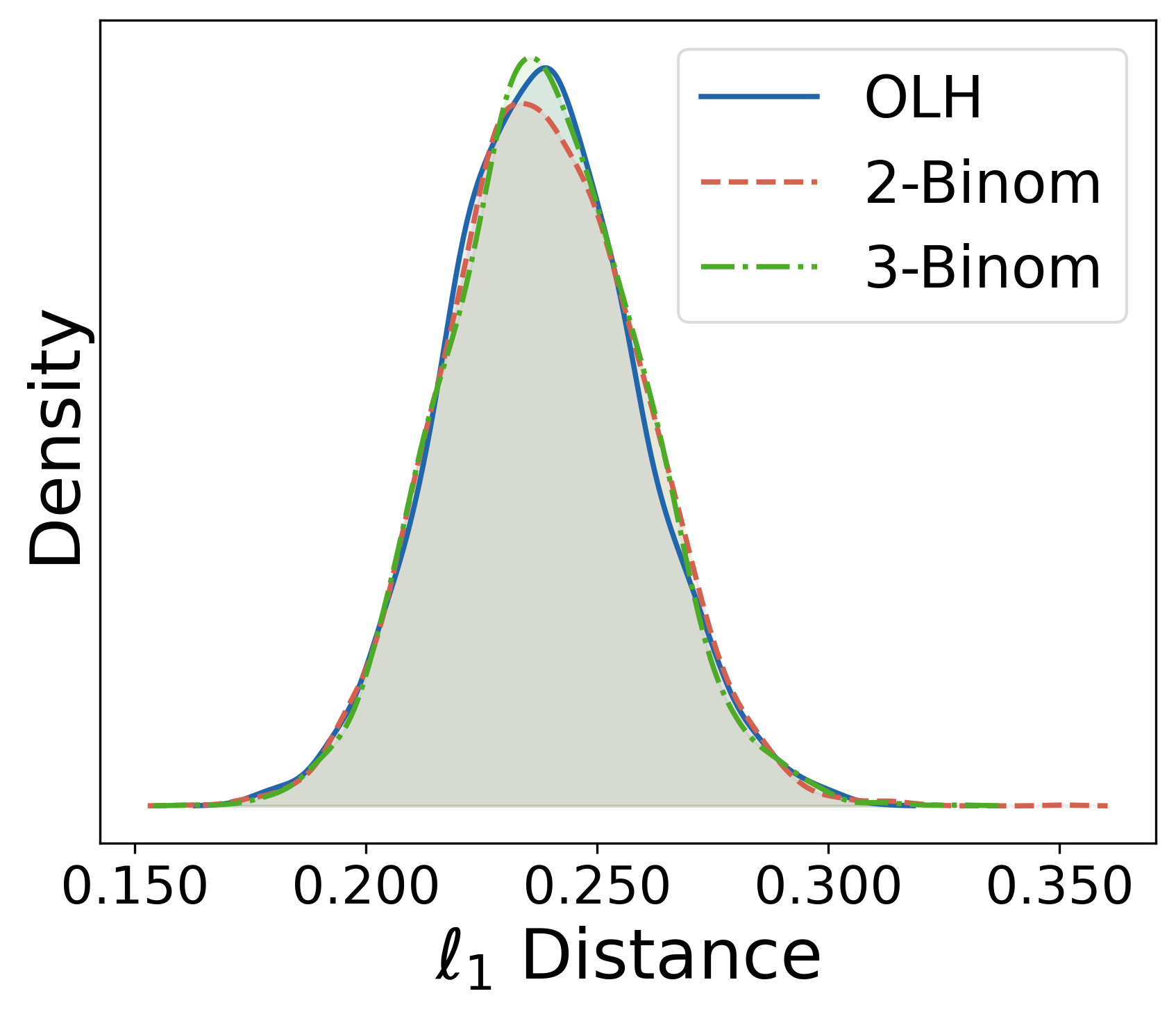}\hfill
    \includegraphics[width=0.24\linewidth]{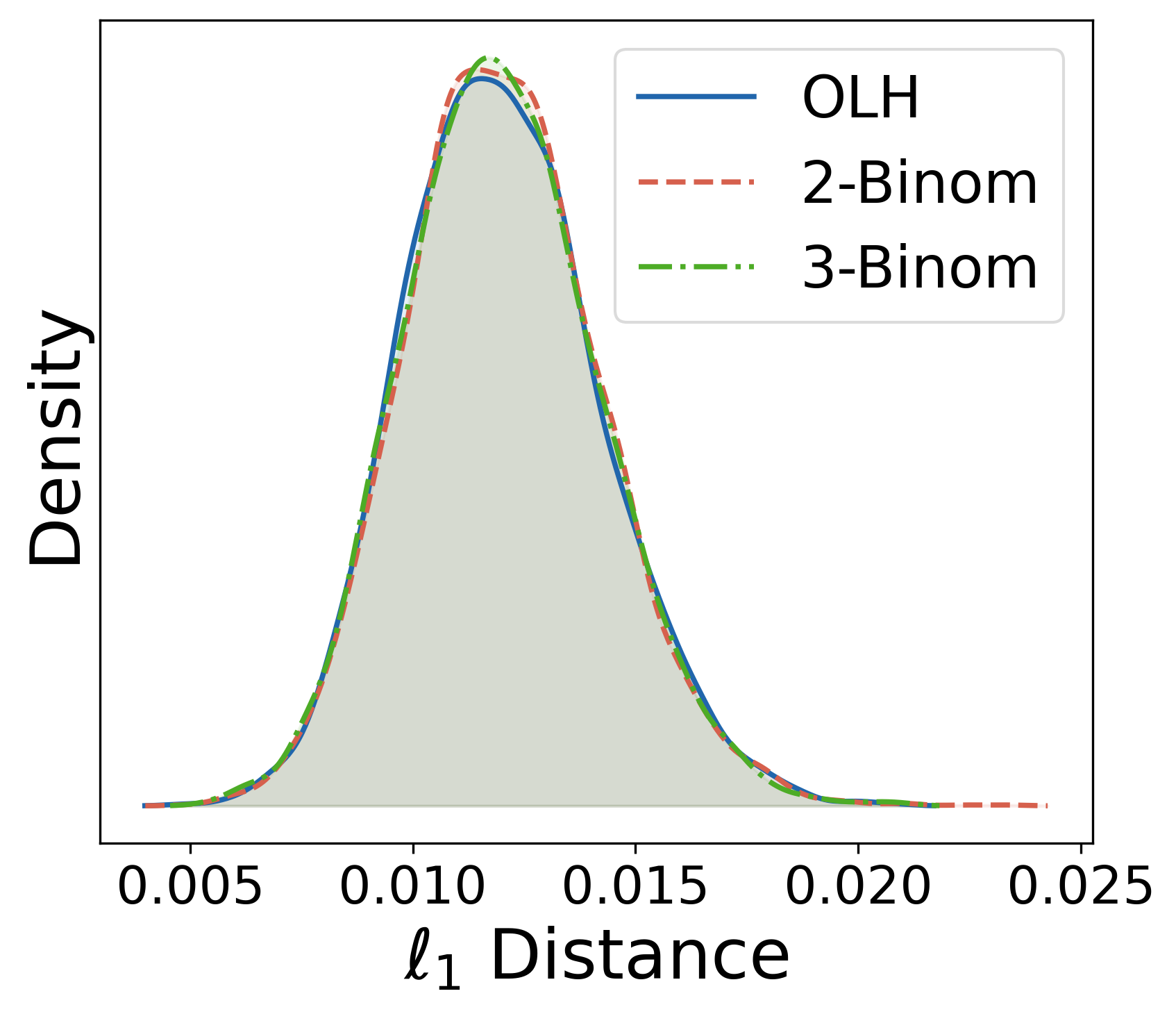}\hfill
    \includegraphics[width=0.24\linewidth]{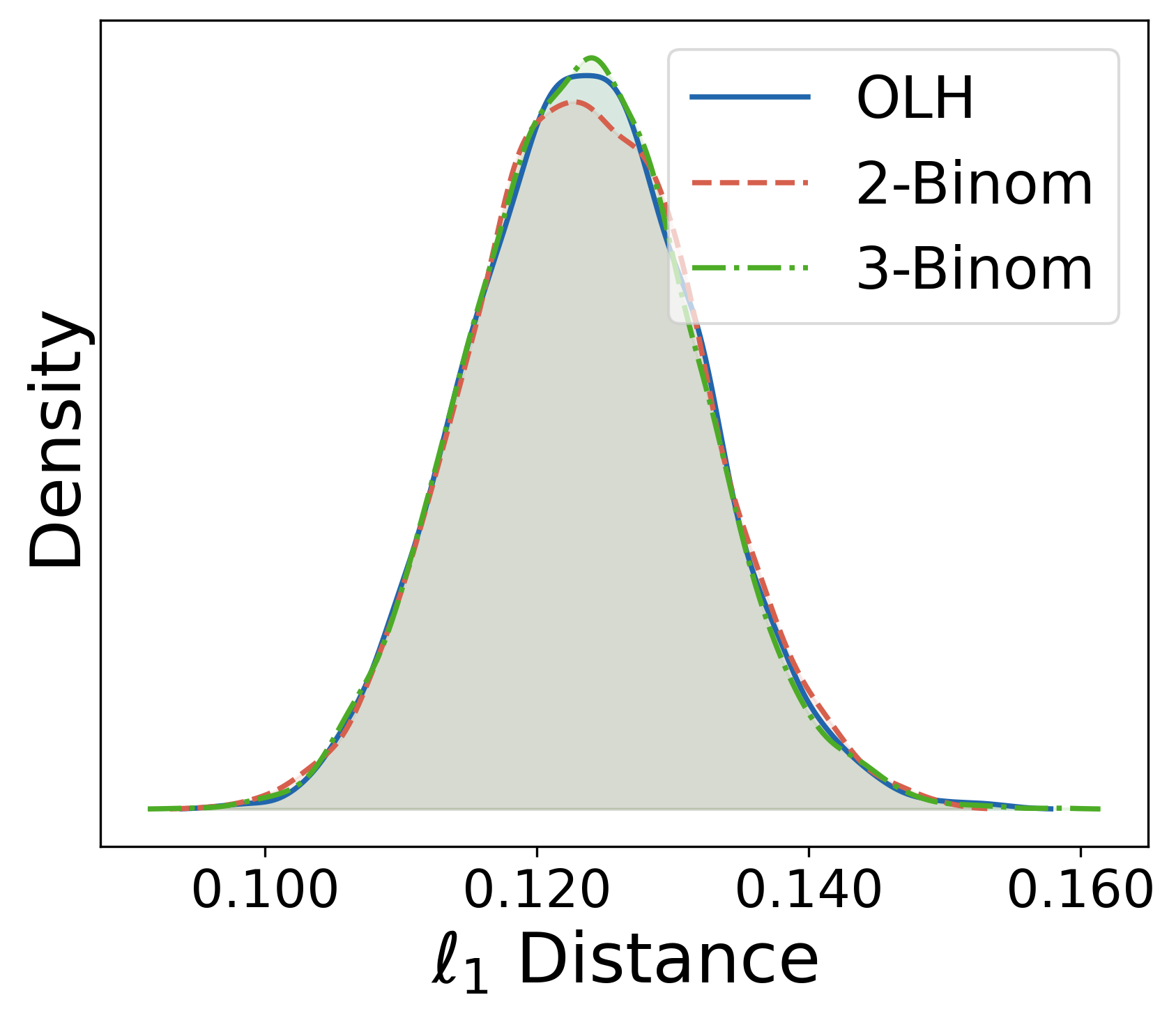}\hfill
    \includegraphics[width=0.24\linewidth]{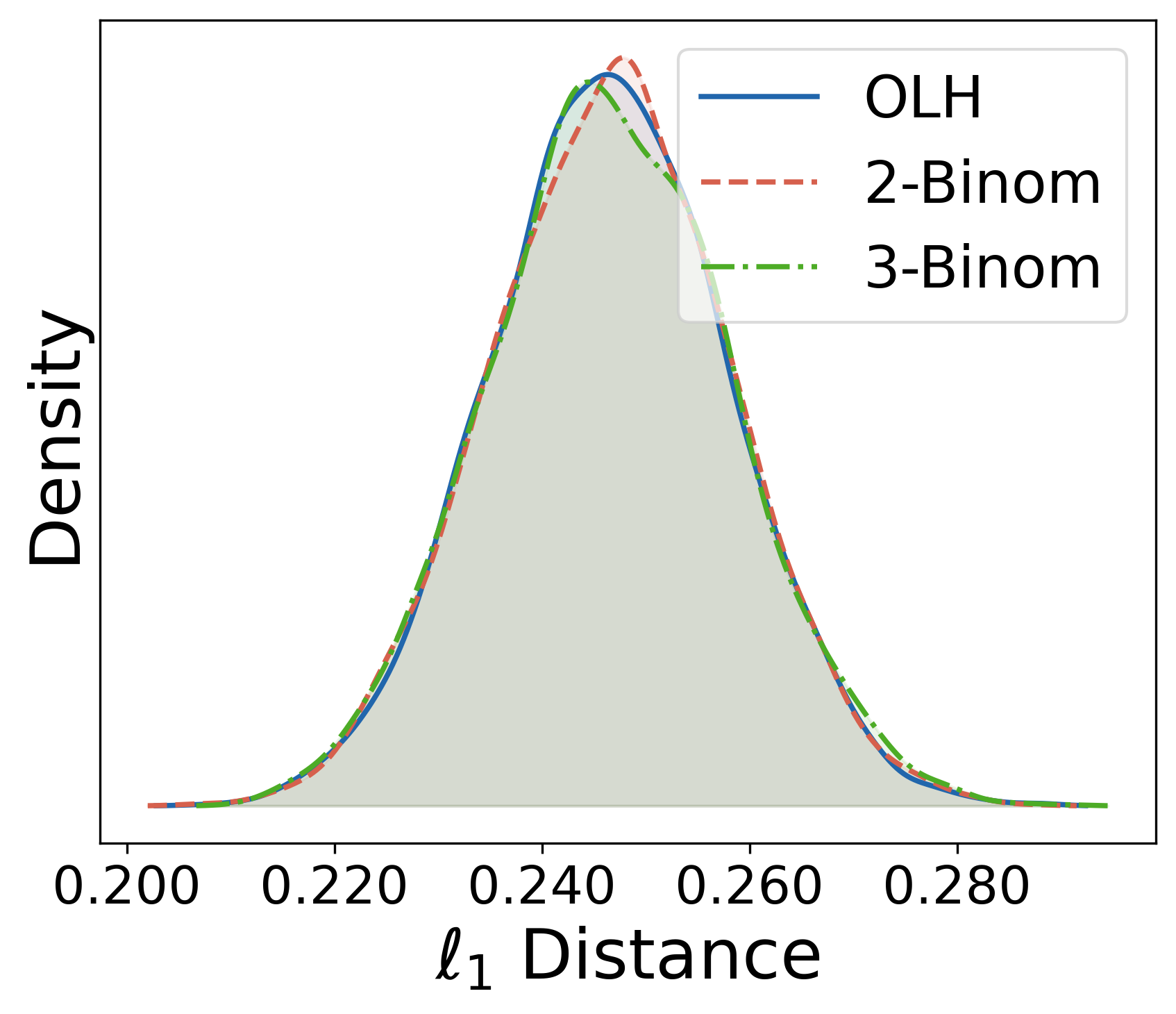}\\[4pt]
    \includegraphics[width=0.24\linewidth]{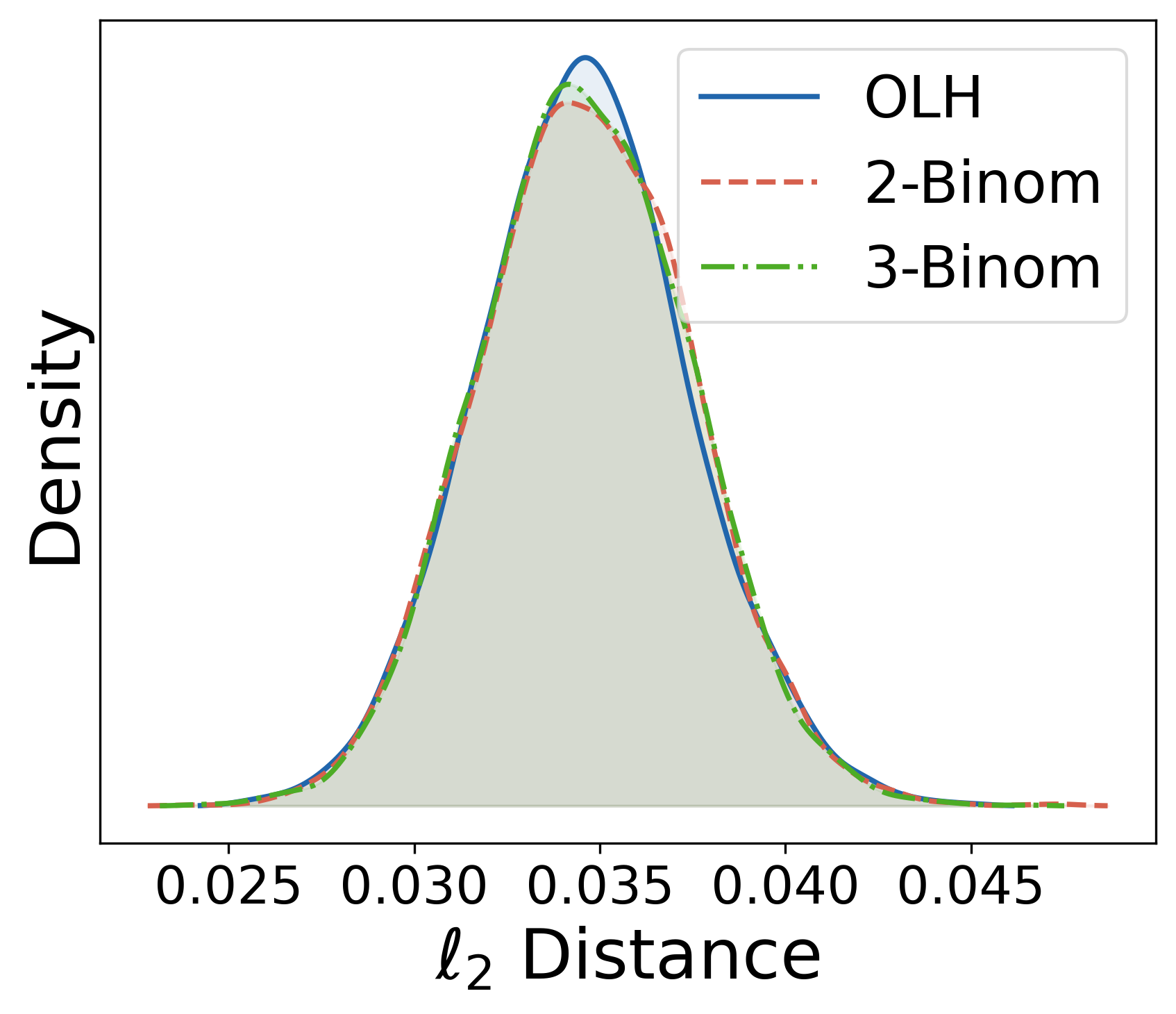}\hfill
    \includegraphics[width=0.24\linewidth]{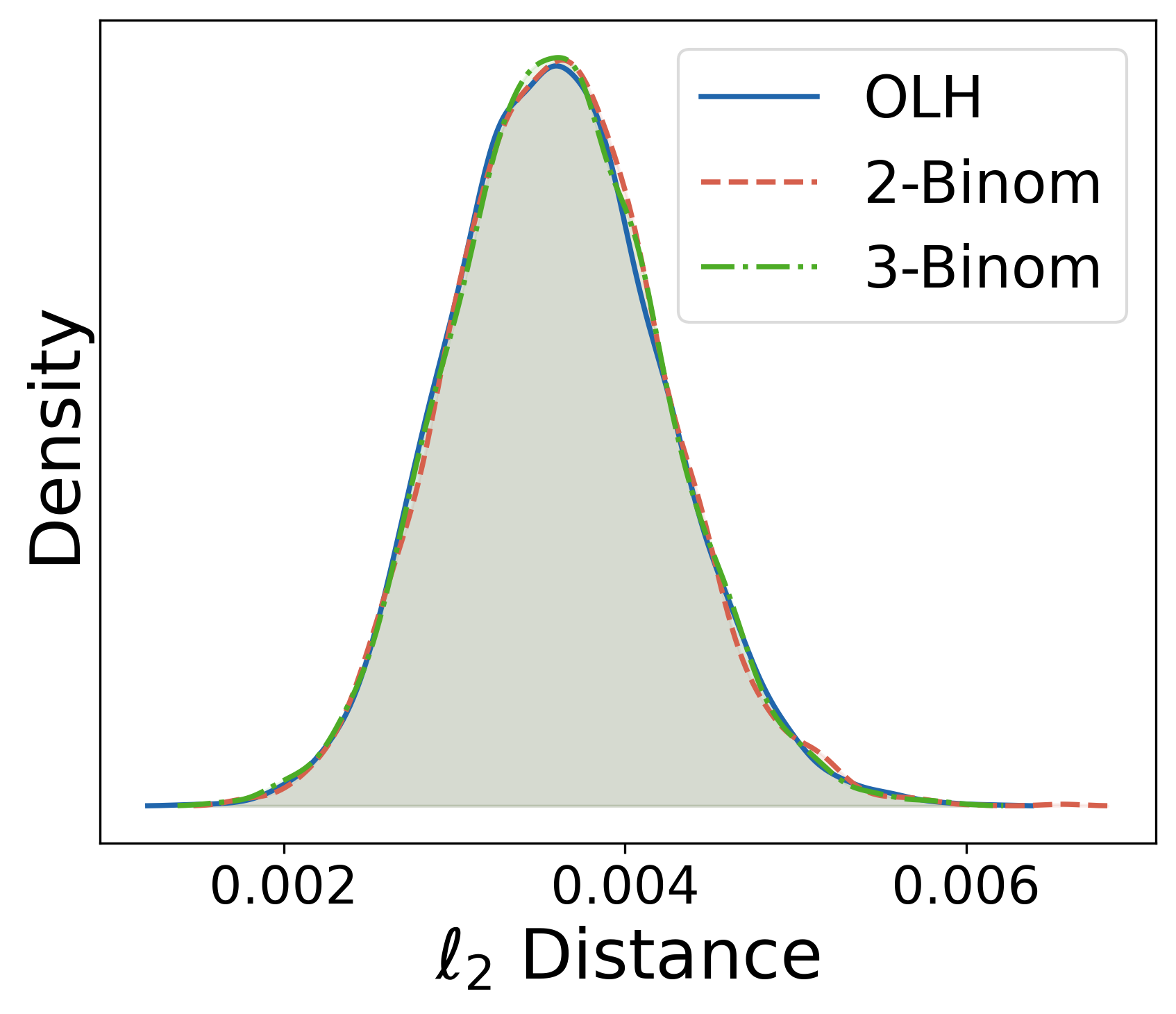}\hfill
    \includegraphics[width=0.24\linewidth]{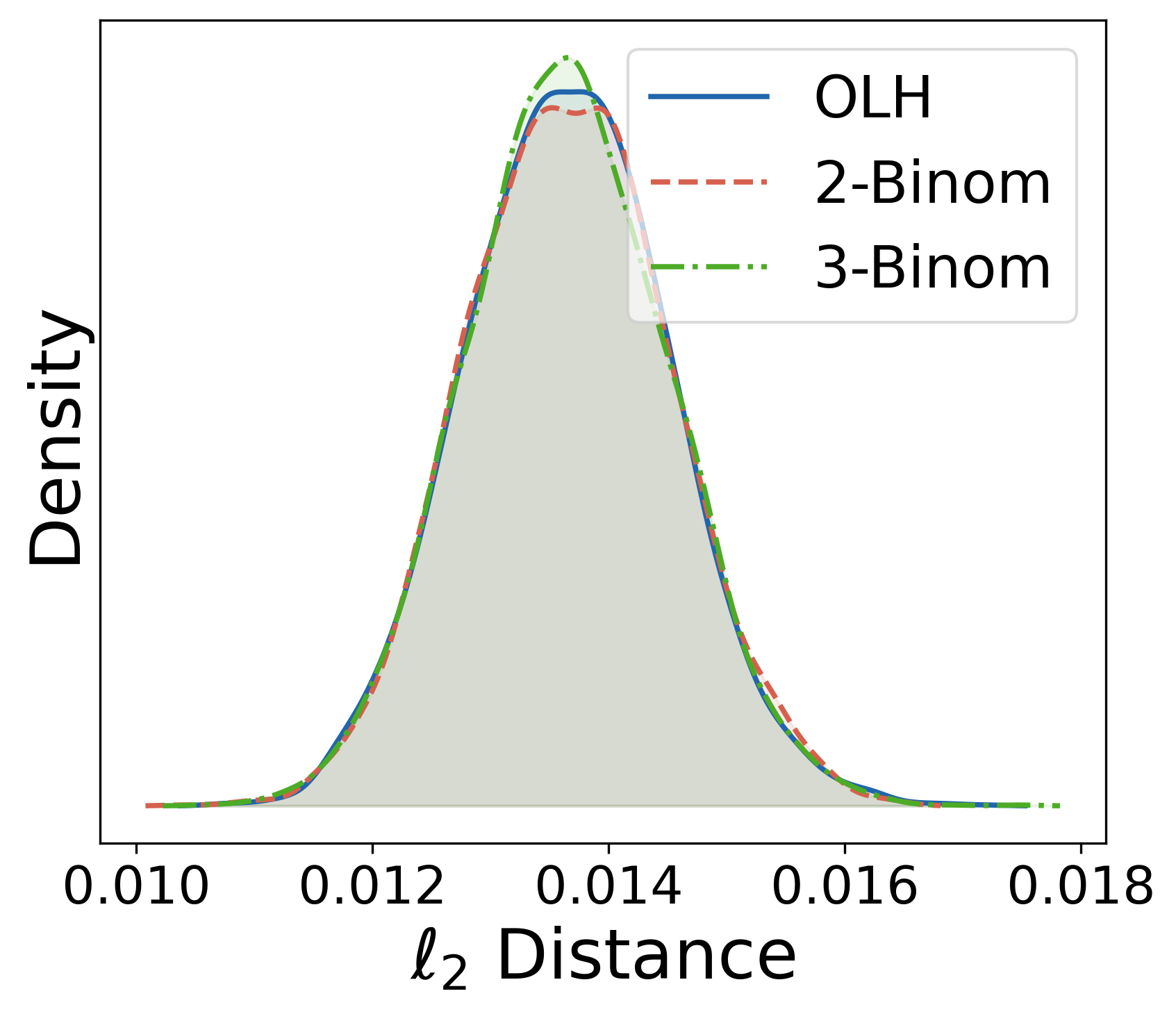}\hfill
    \includegraphics[width=0.24\linewidth]{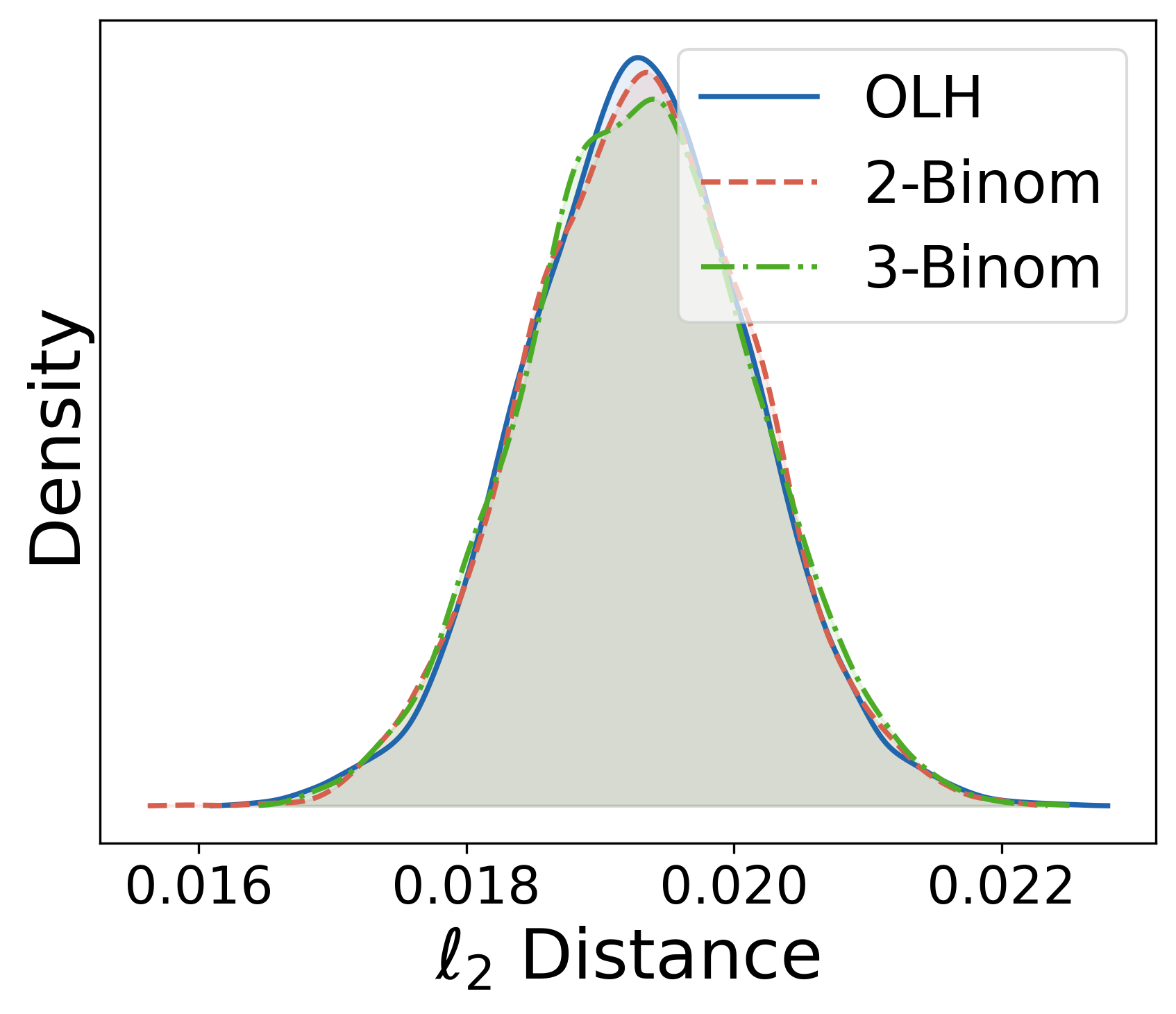}
    \vspace{-8pt}
    \caption{KDE of $\ell_1$ and $\ell_2$ distances across 5{,}000 repetitions at $\varepsilon = 2.0$. Datasets from left to right: Adult, MSNBC, Kosarak, BMS-POS.}
    \label{fig:kde}
\end{figure}

\subsection{Comparison of Empirical Variance}

Finally, we verify that OLH, 2-Binom, and 3-Binom agree in terms of per-value variances. We perform 5{,}000 repetitions for each method at $\varepsilon = 2.0$ and compute the empirical variance of $\hat{f}_v$ for each $v \in \mathcal{D}$. The results are shown in Figure \ref{fig:variance}, where domain values on the $x$-axis are sorted in descending order of their true frequency $f_v$, so that high-frequency values appear on the left and low-frequency values on the right.
 
\begin{figure}[t]
    \centering
    \includegraphics[width=0.24\linewidth]{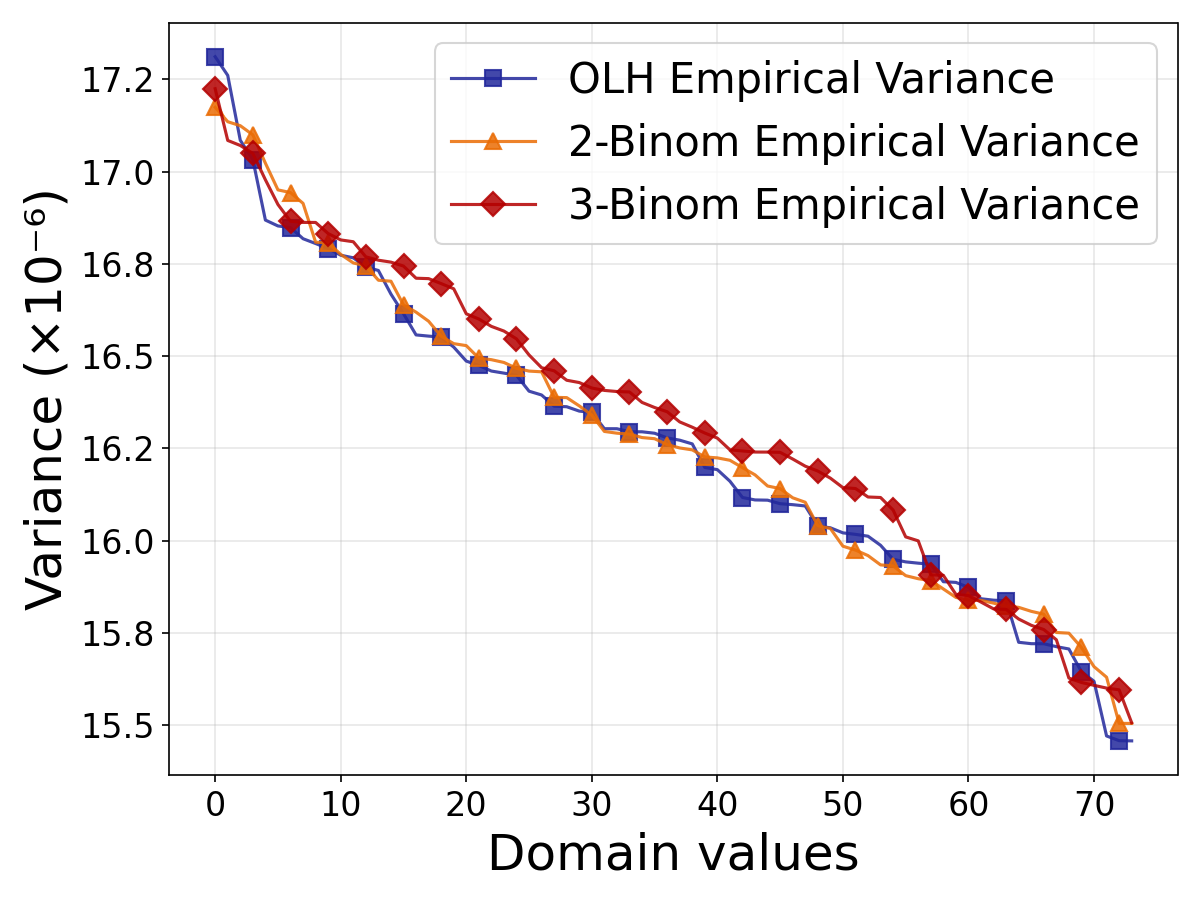}\hfill
    \includegraphics[width=0.24\linewidth]{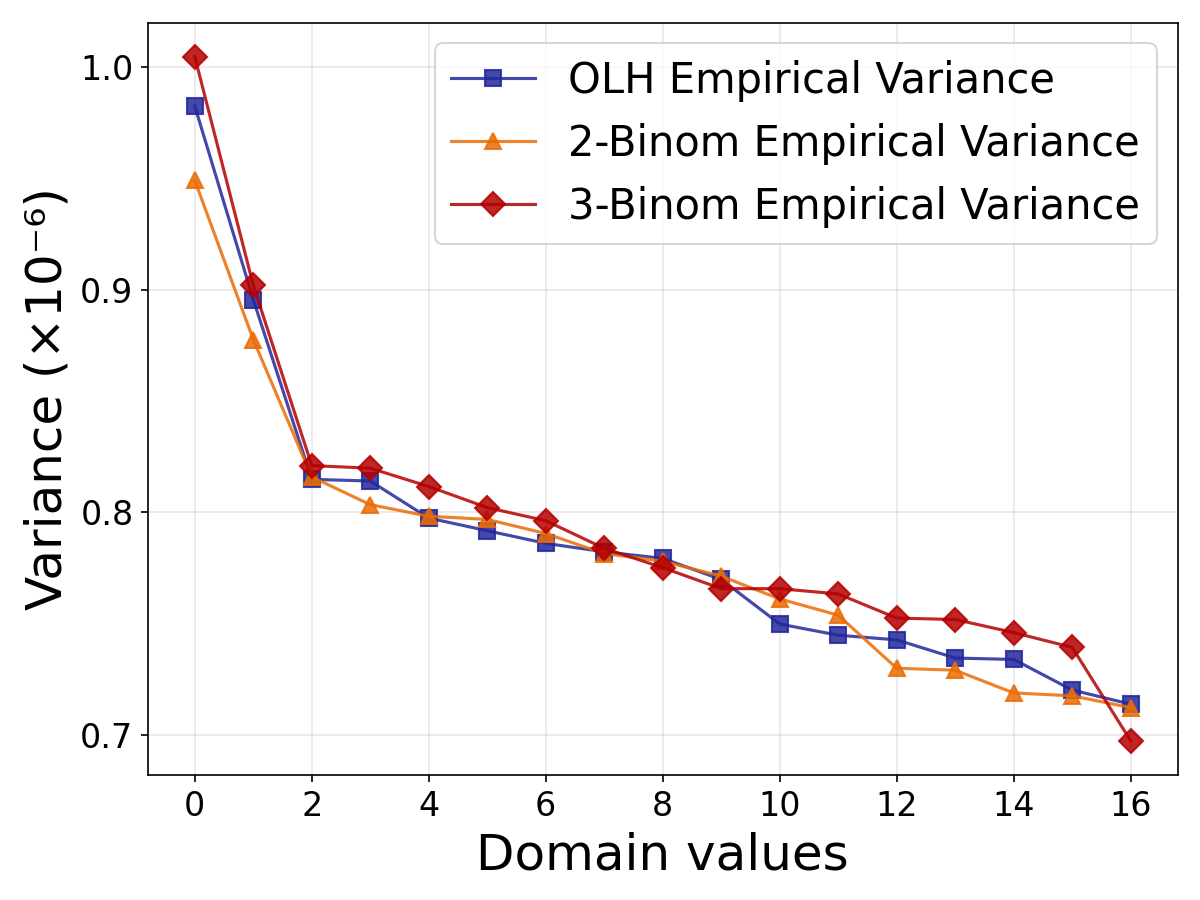}\hfill
    \includegraphics[width=0.24\linewidth]{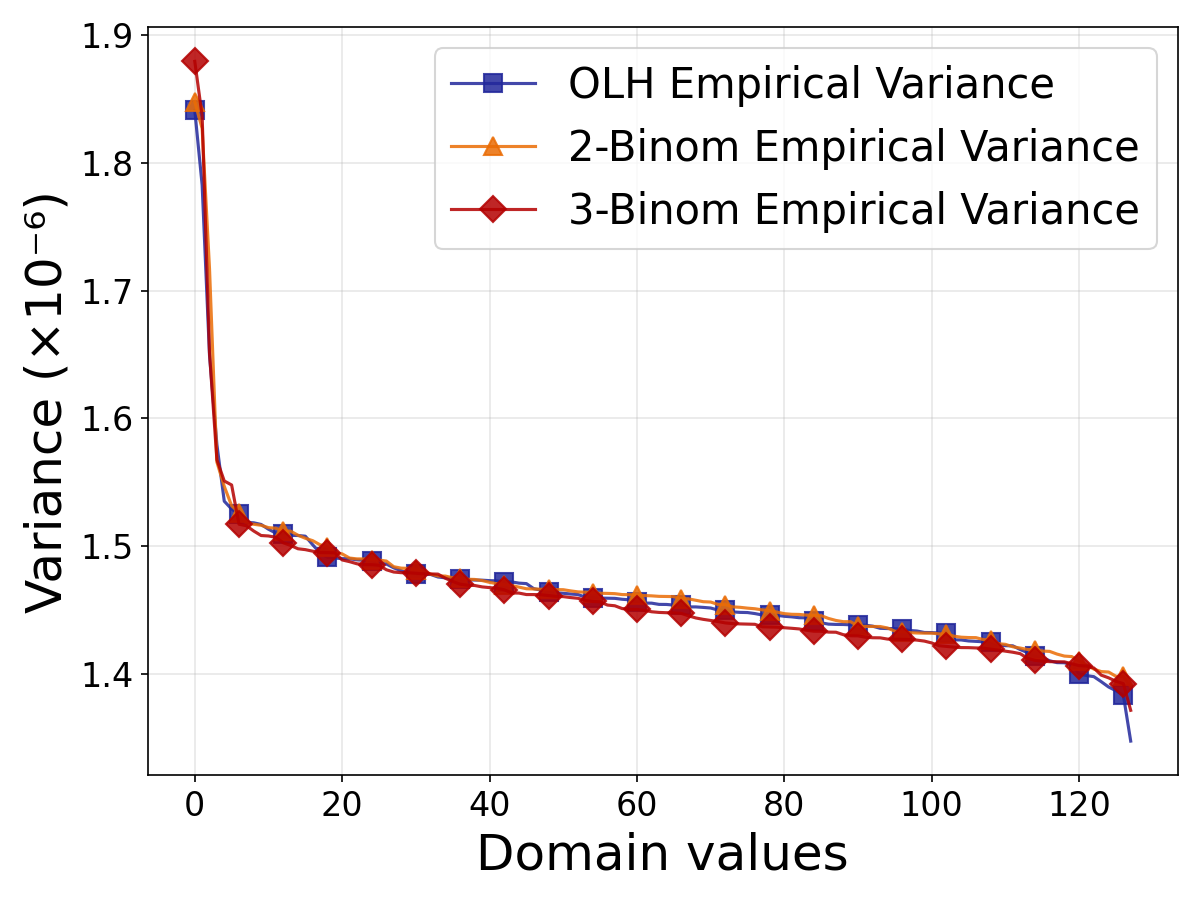}\hfill
    \includegraphics[width=0.24\linewidth]{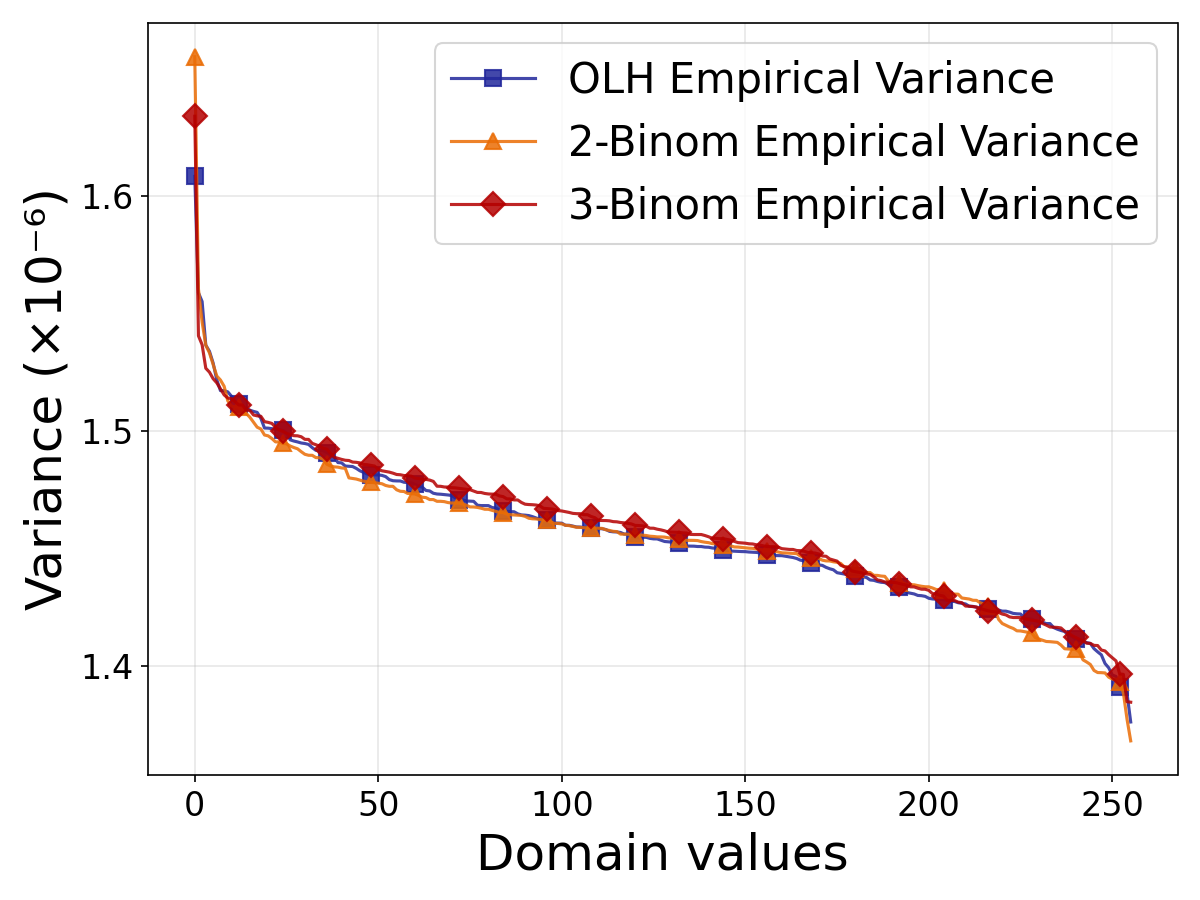}
    \vspace{-8pt}
    \caption{Empirical variances of OLH, 2-Binom, and 3-Binom at $\varepsilon = 2.0$. Datasets from left to right: Adult, MSNBC, Kosarak, BMS-POS.}
    \label{fig:variance}
\end{figure}
 
Across all four datasets, the empirical variance curves of OLH, 2-Binom, and 3-Binom track each other closely and follow the same decreasing trend. This is consistent with our theoretical result in Section \ref{sec:utility}, and the fact that variance is positively correlated with $f_v$. We note that the shapes of the variance curves may differ across different datasets, which is caused by the frequency distributions of the datasets. On MSNBC, the variance drops sharply from index 0 to index 2 and then flattens into a nearly horizontal plateau, since the dataset has two dominant categories whose high $f_v$ causes a distinctly elevated variance for those values. On Kosarak and BMS-POS, the initial drop is steep but the curves decay more gradually over the wider domain, reflecting the long-tail frequency distribution of the datasets. On Adult, the decay is more gradual and uniform, consistent with the smoother age distribution in that dataset.
 
Overall, in all cases, the empirical variance curves of 2-Binom and 3-Binom closely follow those of OLH. This confirms that both methods preserve not only the mean estimation accuracy of OLH but also its per-value variances, providing empirical validation of the theoretical guarantees established in Section~\ref{sec:utility}.

\vspace{-4pt}
\section{Related Work}

LDP has been an active area of research \cite{cormode2018privacy,yang2023local}, with several frequency estimation protocols proposed and studied in the literature, including Generalized Randomized Response (GRR) \cite{kairouz2016discrete}, RAPPOR \cite{erlingsson2014rappor}, Optimized Unary Encoding (OUE), and Optimized Local Hashing (OLH) \cite{wang2017locally}. Among them, OLH is particularly attractive for large $d$ due to its per-user communication cost and optimized variance, and it has become one of the most widely adopted protocols in the LDP literature. For example, Wang et al.~\cite{wang2018locally,wang2019locally} used OLH in frequent itemset mining and heavy hitter identification. Cormode et al.~\cite{cormode2019answering} proposed methods for hierarchical range query answering via LDP frequency estimators such as OLH. Xu et al.~\cite{xu2020collecting} used OLH for analyzing data jointly from multiple services under LDP. Qian et al.~\cite{qian2023collaborative} used OLH for multi-dimensional value collection under LDP. Xu et al.~\cite{xu2023mlpkv} proposed MLPKV for private key-value data collection using OLH.

There were also studies conducted regarding the fairness and adversarial robustness aspects of LDP protocols, including OLH. Gursoy et al.~\cite{gursoy2022adversarial} analyzed protocols, including OLH, from a Bayesian adversarial perspective. Arcolezi et al.~\cite{arcolezi2023risks} demonstrated the risks of collecting multidimensional data under LDP and used OLH as one of the protocols. Cao et al.~\cite{cao2021data} demonstrated data poisoning attacks on LDP protocols, and used OLH as one of the protocols. Balioglu et al.~\cite{balioglu2025don} specifically studied OLH by focusing on the privacy and poisoning-related impacts of hash function choice. Many other works can also be found which utilize OLH; thus, many works would benefit from our contribution of increasing the speed of OLH simulations.

Several simulation and benchmarking platforms for LDP protocols have been developed. Cormode et al.~\cite{cormode2021frequency} introduced \textsc{Pure-LDP}, a Python benchmarking package implementing OLH, OUE, and other protocols. Arcolezi et al.~\cite{arcolezi2022multi} developed \textsc{Multi-Freq-LDP}, a toolkit supporting one-time, multidimensional, and longitudinal frequency estimation with local hashing. Arcolezi and Gambs~\cite{arcolezi2024revealing} built LDP-Auditor, a framework for empirically auditing the privacy loss of LDP protocols through distinguishability attacks across eight frequency estimation protocols. Khodaie et al.~\cite{khodaie2025postprocessing} introduced a benchmarking platform evaluating LDP protocols and post-processing methods across multiple datasets and utility metrics. The existing OLH implementations in all of these platforms follow the standard $O(nd)$ simulation approach; our proposed method can serve as a drop-in replacement to improve their execution times.
 
Finally, our work builds on Karatas and Gursoy~\cite{karatas2026accelerating}, who proposed Binomial modeling for bitvector-based LDP protocols (RAPPOR and OUE), reducing their simulation complexity from $O(nd)$ to $O(n+d)$. Extending their approach to OLH is non-trivial, as the hash-based structure of OLH introduces a qualitatively different perturbation mechanism: users whose true value differs from $v$ can contribute to $\mathrm{Sup}(v)$ through either a hash collision or a perturbation event, requiring a distinct Binomial decomposition. 

\section{Conclusion} \label{sec:conclusion}

In this paper, we addressed the $O(nd)$ simulation cost of OLH, a widely used hash-based LDP protocol, by introducing two Binomial-based simulation algorithms: 2-Binom and 3-Binom. Our algorithms replace user-by-user simulation with draws from a small number of Binomial random variables, reducing complexity to $O(n+d)$ and potentially further to $O(d)$ per repetition when counts are precomputed. We formally proved that both algorithms yield unbiased frequency estimates and carry variance identical to that of the original OLH simulation. Our algorithms were also validated empirically using four real-world datasets and a range of $\varepsilon$ budgets. Results show that 2-Binom and 3-Binom match the estimation results of standard OLH simulation while reducing execution times from several minutes to milliseconds. In future work, we plan to incorporate 2-Binom and 3-Binom into LDP benchmarking platforms and apply the speedup to downstream OLH-based applications.

\section*{Acknowledgments}
\vspace{-2pt}

This work was supported by the Scientific and Technological Research Council of Türkiye (TUBITAK) under grant number 123E179 and the BAGEP Outstanding Young Scientist Award. The authors thank TUBITAK and the Science Academy for their support.

\bibliographystyle{splncs04}
\bibliography{references}

\appendix

\section{Simplification of the Variance Bracket in Theorem~\ref{thm:variance-three}}
\label{appendix:bracket}

We show that substituting $p = \frac{e^\varepsilon}{e^\varepsilon+g-1}$ and $q = \frac{1}{e^\varepsilon+g-1}$ into the bracket:
\begin{equation}
    p(1-p) + (g-1)q(1-q) + (p-q)^2\left(1-\frac{1}{g}\right)
\end{equation}
yields $1 - \frac{1}{g}$. Let $\beta = e^\varepsilon + g - 1$ for simplicity, so that $p = \frac{e^\varepsilon}{\beta}$ and $q = \frac{1}{\beta}$. We expand each term individually. First expand $p(1-p)$:
\begin{equation}
    p(1-p) = \frac{e^\varepsilon}{\beta} - \frac{e^{2\varepsilon}}{\beta^2} = \frac{e^\varepsilon \beta - e^{2\varepsilon}}{\beta^2} = \frac{e^\varepsilon(\beta - e^\varepsilon)}{\beta^2} = \frac{e^\varepsilon(g-1)}{\beta^2}
\end{equation}
where we used $\beta - e^\varepsilon = g - 1$. Next, expand $(g-1)q(1-q)$:
\begin{equation}
    (g-1)q(1-q) = \frac{g-1}{\beta} - \frac{g-1}{\beta^2} = \frac{(g-1)(\beta-1)}{\beta^2} = \frac{(g-1)(e^\varepsilon+g-2)}{\beta^2}
\end{equation}
where we used $\beta - 1 = e^\varepsilon + g - 2$. Next, expand $(p-q)^2\left(1 - \frac{1}{g}\right)$:
\begin{equation}
    p - q = \frac{e^\varepsilon - 1}{\beta}, \qquad (p-q)^2 = \frac{(e^\varepsilon-1)^2}{\beta^2}
\end{equation}
\begin{equation}
    (p-q)^2\left(1-\frac{1}{g}\right) = \frac{(e^\varepsilon-1)^2}{\beta^2} \cdot \frac{g-1}{g}
\end{equation}
Sum all three terms over the common denominator $\beta^2$, factor out $(g-1)$ and then simplify the inner bracket:
\begin{align}
    &\frac{1}{\beta^2}\left[e^\varepsilon(g-1) + (g-1)(e^\varepsilon+g-2) + \frac{(e^\varepsilon-1)^2(g-1)}{g}\right]  \\
    &= \frac{g-1}{\beta^2}\left[e^\varepsilon + (e^\varepsilon+g-2) + \frac{(e^\varepsilon-1)^2}{g}\right] \label{eq:47} 
\end{align}
Simplify the inner bracket:
\begin{align}
    e^\varepsilon + (e^\varepsilon+g-2) + \frac{(e^\varepsilon-1)^2}{g} &= 2e^\varepsilon + g - 2 + \frac{e^{2\varepsilon} - 2e^\varepsilon + 1}{g} \\
    &= \frac{g(2e^\varepsilon + g - 2) + e^{2\varepsilon} - 2e^\varepsilon + 1}{g} \\
    &= \frac{e^{2\varepsilon} + 2(g-1)e^\varepsilon + (g-1)^2}{g}
\end{align}
Recognizing that the numerator is a perfect square:
\begin{equation}
    = \frac{(e^\varepsilon + g - 1)^2}{g} = \frac{\beta^2}{g}
\end{equation}
Substitute this back to Eq.~\ref{eq:47}:
\begin{equation}
    \frac{g-1}{\beta^2} \cdot \frac{\beta^2}{g} = \frac{g-1}{g} = 1 - \frac{1}{g}
\end{equation}
which is the result claimed in the proof of Theorem~\ref{thm:variance-three}.
 
\end{document}